\documentclass[journal, oneside]{IEEEtran}

\IEEEoverridecommandlockouts

\usepackage{amssymb,amsmath,amsthm, amsfonts}
\usepackage{algorithm}
\usepackage{algpseudocode}
\usepackage{array}
\usepackage{graphicx}
\usepackage{cite}
\usepackage{subcaption}
\usepackage[inline]{enumitem}
\usepackage{soul}
\usepackage{caption}
\usepackage{xcolor}

\DeclareMathOperator*{\argmin}{arg\,min}
\usepackage{lipsum}
\usepackage{balance}

\allowdisplaybreaks

\newcommand{\tp}{\mathsf{T}}
\newcommand{\tpc}{\mathsf{H}}
\newcommand{\ids}{\mathcal{I}}
\newcommand{\ar}{\mathcal{R}}
\newcommand{\ex}{\mathbb{E}}
\newcommand{\pr}{\mathrm{Pr}}
\newtheoremstyle{italiclabel}
{\topsep}   
{\topsep}   
{\normalfont}  
{}  
{\itshape}  
{.}  
{.5em}  
{}  

\theoremstyle{italiclabel}

\newtheorem{definition}{Definition}

\newtheorem{proposition}{Proposition}

\theoremstyle{normalfont}

\begin{document}
	
	\title{\huge Region-Based Constellation Designs for \\Constructive Interference Precoding in MU-MIMO}
	
	\author{Yupeng Zheng, Chunmei Xu, Jinfei Wang, Yi Ma, and Rahim Tafazolli
	}
	\markboth{}%
	{}
	
	
	\maketitle

\begin{abstract}
	The performance of constructive interference precoding (CIP) for multi-user multi-antenna (MU-MIMO) systems is governed by the structure of the constructive interference (CI) regions, yet this is overlooked in conventional constellation design. This work proposes the region-based constellation (RBC) model to lay the foundation for CIP constellation design. An RBC directly defines the mapping between messages and their feasible regions, instead of deriving them from an existing constellation. To provide insight for RBC design, we study the limitations of quadrature-amplitude-modulation (QAM)-based CIP. Analytical results show that the restrictive CI regions of QAM symbols are systematically misaligned with the objective-minimising sign pattern, resulting in a significant gap to the theoretical performance limit. From the perspective of improving sign alignment, two novel RBC schemes with non-convex feasible regions are proposed, namely mirrored-ends QAM (ME-QAM) and real-extended ME-QAM. A low-complexity algorithm is also developed for the resulting mixed-integer quadratic program, achieving a complexity comparable to QAM-based CIP. Simulation results with constellation sizes $\{16,64\}$ demonstrate up to $4$~dB signal-to-noise-ratio gain of the proposed schemes over QAM-based CIP. The proposed RBC model is also applicable to other systems with non-bijective modulation, representing a promising direction for future research.
\end{abstract}
\begin{IEEEkeywords}
	Constructive interference precoding, multi-user MIMO, region-based constellation, symbol-level precoding, constellation design.
\end{IEEEkeywords}

\section{Introduction}
Constellation designs for digital modulation have been extensively studied in the literature, with research spanning optimal geometric shaping, probabilistic shaping, and labelling to maximise spectral efficiency and error performance~\cite{Forney1989,Forney1992TrellisShaping,Caire1998,Agrell2004,Bocherer2015,Yao2020}. However, these works assume a bijective mapping between constellation points and the information they carry, which does not apply to systems where the value of a data-bearing symbol can be dynamically selected from a feasible set. This class of systems includes multi-user multi-antenna (MU-MIMO) downlink with constructive interference precoding (CIP)~\cite{Masouros2011,Li2020tutorial,Wang2024}, MU-MIMO downlink with vector perturbation precoding~\cite{Hochwald2005,Ma2016,Wang2022}, orthogonal frequency-division multiplexing (OFDM) with active constellation extension~\cite{Krongold2003,Lin2021}, and OFDM with tone injection~\cite{Tellado1998,Jacklin2013}. This paper focuses on constellation design for MU-MIMO downlink with CIP.

MU-MIMO plays a central role in modern wireless communications, enabling significant spectral efficiency gains through spatial multiplexing~\cite{Tse2005}. Precoding exploits channel state information (CSI) at the base station (BS) to simultaneously transmit information to multiple users while managing inter-user interference~\cite{Gesbert2007}. 
CIP has emerged as a promising precoding technique that shapes interference to be constructive for detection rather than eliminating it~\cite{Li2020tutorial}. Compared to conventional linear precoding such as zero-forcing (ZF) and regularized ZF~\cite{Peel2005}, CIP significantly improves reliability without sacrificing spectral efficiency.
Originally, CIP optimizes the precoding matrix to retain only the interference that pushes the received symbols toward their correct decision regions~\cite{Masouros2011}. Recent works adopt the more tractable symbol-perturbation framework, where a non-bijective modulation is followed by a channel-dependent linear precoder~\cite{Liu2022,Wang2022weighted}. 

Conventionally, an existing constellation such as quadrature amplitude modulation (QAM) or phase-shift keying (PSK) is first selected as the basis for the non-bijective modulation in CIP. Then, each constellation point is extended to a feasible region, i.e., the constructive interference (CI) region~\cite{Haqiqatnejad2018}. This procedure has two major limitations. First, the geometry of the CI regions governs the constraint structure of the CIP optimization and hence its performance, yet this is not accounted for in conventional constellation design. For instance, QPSK, which has highly flexible CI regions, achieves approximately $15$~dB signal-to-noise-ratio (SNR) gain over ZF precoding~\cite{Li2018}, while the gain drops to $5$~dB for $64$-QAM with its more restrictive CI regions~\cite{Li2020interference}. Second, the CI regions are limited to convex sets by construction, which is an unnecessary restriction for CIP. 
Our prior work~\cite{Zheng2025} shows that introducing non-convex feasible regions to a QAM constellation can improve the performance of CIP in reconfigurable-intelligent-surface-enhanced MU-MIMO systems.

The main contributions are summarized as follows:
\begin{enumerate}[leftmargin=1.6em,labelsep=0.4em,itemsep=0pt,topsep=0pt]
	\item We perform a novel analysis of the QAM-based CIP optimization problem from the perspective of the sign pattern in the constraints. It is shown analytically that only $1/2$ of the exploitable degrees of freedom (DoF) are aligned with their objective-minimizing sign pattern on average. This misalignment systematically limits the performance of QAM from the analytical bound in CIP. Moreover, a modification which introduces sign flexibility into the constraints increases this proportion to~$3/4$.
%
	\item A novel \emph{region-based constellation (RBC)} model is proposed to lift the restrictions of the CI-region-based CIP. Instead of deriving the CI regions from an existing constellation, this model directly describes the mappings between messages and their feasible regions. Enabled by this model, two novel constellation schemes are proposed: \emph{mirrored-ends QAM (ME-QAM)} and \emph{real-extended ME-QAM (RM-QAM)}. Non-convex feasible regions are applied in both schemes to enhance the sign-alignment capability. RM-QAM further introduces unconstrained DoF to obtain additional flexibility.
	\item A low-complexity algorithm is proposed to solve the resulting non-convex optimization problem with complexity comparable to QAM-based CIP. The non-convex feasible regions transform the CIP problem from a linearly-constrained quadratic program (LCQP) to a mixed-integer QP (MIQP). Solving the MIQP to global optimality requires exhaustive search over all sign patterns, incurring exponential complexity. In contrast, the proposed algorithm predicts the sign pattern via a closed-form expression.
	\item Simulation results of symbol error rate (SER) demonstrate that the proposed schemes achieve up to $4$~dB gain over QAM in CIP. Notably, ME-QAM achieves a $1$~dB gain in a regime where QAM-based CIP offers no benefit over ZF. Results under imperfect CSI confirm that the SER advantage is consistent across moderate levels of channel estimation error. Block error rate (BLER) results with channel coding additionally verify that the SER gains translate to coded performance. 
\end{enumerate}

The rest of this paper is organized as follows. Section~\ref{sec2} introduces the MU-MIMO system model and the non-bijective modulation in CIP. Section~\ref{sec3} revisits the QAM-based CIP. Section~\ref{sec4} presents the proposed RBC model, ME-QAM and RM-QAM schemes, and algorithms for MIQPs. Simulation results are presented and discussed in Section~\ref{sec5}. Section~\ref{sec6} concludes the paper. 

\textit{Notation:} Boldface lowercase and uppercase letters denote vectors and matrices, respectively, e.g., $\mathbf{a}$ and $\mathbf{A}$. Calligraphic uppercase letters denote sets, e.g., $\mathcal{C}$. $\mathbb{R}$ and $\mathbb{C}$ denote the sets of all real and complex numbers, respectively. $\Re(\cdot)$ and $\Im(\cdot)$ denote the real and imaginary parts of a complex number, respectively. $(\cdot)^\mathsf{T}$, $(\cdot)^\mathsf{H}$, $(\cdot)^*$, and $(\cdot)^\dagger$ denote the transpose, conjugate transpose, complex conjugate, and Moore--Penrose pseudoinverse, respectively. $\mathbf{I}_n$ denotes the $n\times n$ identity matrix, and $\mathbf{1}_n$ denotes the $n$-dimensional all-one vector. $\|\cdot\|$ denotes the Euclidean norm and $|\cdot|$ denotes the element-wise absolute value of a scalar or the cardinality of a set. $\mathbb{E}(\cdot)$ denotes the expectation operator and $\mathrm{Pr}\{\cdot\}$ denotes the probability of an event. $\mathcal{N}(0,\sigma^2)$ and $\mathcal{CN}(0,\sigma^2)$ denote the zero-mean real Gaussian and circularly symmetric complex Gaussian distributions with variance $\sigma^2$, respectively. $\operatorname{sgn}(\cdot)$ denotes the sign function. $\operatorname{diag}(\cdot)$ denotes the diagonal matrix formed from a vector. $\mathrm{tr}(\cdot)$ denotes the trace of a matrix. $\lfloor\cdot\rfloor$ and $\lfloor\cdot\rceil$ denote rounding down to and rounding to the nearest integer, respectively.

\section{System Model}\label{sec2}
\subsection{MU-MIMO Downlink with Linear Precoding}
We consider a MU-MIMO downlink system. A BS equipped with $N_{\mathrm{t}}$ transmit antennas simultaneously serves $K$ single-antenna users, with $K \leq N_{\mathrm{t}}$. The wireless channel is modeled as flat fading, which is standard for narrowband or OFDM per-subcarrier transmission. The channel between the $n$th transmit antenna and the $k$th user is denoted as $h_{k,n}$. This paper assumes the i.i.d. Rayleigh fading model with each $h_{k,n}$ independently following
\begin{equation}\label{eq:raychan}
	h_{k,n} \sim \mathcal{CN}(0,1), \quad \forall k \in \{1,\ldots,K\},\ n \in \{1,\ldots,N_\mathrm{t}\},
\end{equation}
which models rich scattering environments and is widely adopted in the MU-MIMO literature for its analytical tractability~\cite{Tse2005}. Unless otherwise stated, the channel is assumed to be perfectly known at the BS.

In each symbol duration, the BS aims to transmit a message $m_k$, $k=1,2,\cdots,K$, to each $k$th user simultaneously. Each $m_k$ is independently drawn from the set $\mathcal{M}=\{0,1,\cdots,M-1\}$ with equal probability, where $M$ denotes the constellation size. In this paper, we are interested in the cases where $M\ge 16$, for which canonical square QAM is the standard modulation scheme widely adopted in modern communication systems. Each message $m_k$ is then modulated into a complex symbol $s_k\in\mathbb{C}$ via a function $f$. Conventionally, $f$ is a bijective mapping from $\mathcal{M}$ to a constellation set $\mathcal{C}=\{c_m\in\mathbb{C}\mid m\in\mathcal{M}\}$.  The vector $\mathbf{s}=[s_1,\cdots,s_K]^\mathsf{T}$ is mapped to the transmit vector $\mathbf{x}\in\mathbb{C}^{N_\mathrm{t}}$, which is the collection of the transmit signal at each BS antenna. With linear precoding, $\mathbf{x}$ is given by
\begin{equation}
	\mathbf{x} = \mathbf{W}\mathbf{s},
\end{equation}
where $\mathbf{W}\in\mathbb{C}^{N_{\mathrm{t}}\times K}$ denotes the linear precoder. The received signal at user $k$ is given by
\begin{equation}\label{eq:yk}
	y_k = \alpha^{-1}\mathbf{h}_k^\mathsf{T}\mathbf{W}\mathbf{s}+v_k,
\end{equation}
where $\mathbf{h}_k = [h_{k,1},\cdots ,h_{k,N_\mathrm{t}}]^\tp \in \mathbb{C}^{N_{\mathrm{t}}}$ denotes the channel vector between the BS and user~$k$, $v_k \sim \mathcal{CN}(0,\sigma^2)$ is the additive white Gaussian noise, and $\alpha = \|\mathbf{Ws}\|$ denotes the rescaling factor which constrains the total transmit power to~$1$. 

Let $\mathbf{H} = [\mathbf{h}_1,\ldots,\mathbf{h}_K]^{\mathsf{T}} \in \mathbb{C}^{K \times N_{\mathrm{t}}}$ denote the aggregated channel matrix. Stacking the received signals of all users, the system input--output relation is given by
\begin{equation}\label{eq:bar_y}
	\mathbf{y} = \alpha^{-1}\mathbf{H}\mathbf{W}\mathbf{s} + \mathbf{v},
\end{equation}
where $\mathbf{y} = [y_1,\ldots,y_K]^{\mathsf{T}}$ and $\mathbf{v} = [v_1,\ldots,v_K]^{\mathsf{T}}$. To focus on the problem of interest, $\mathbf{W}$ is selected as the ZF precoder, which is given by the Moore-Penrose pseudoinverse as
\begin{equation}\label{eq:zf}
	\mathbf{W}=\mathbf{H}^\dagger=\mathbf{H}^\mathsf{H}(\mathbf{H}\mathbf{H}^\mathsf{H})^{-1}.
\end{equation}
By substituting \eqref{eq:zf} into \eqref{eq:yk}, inter-user interference is perfectly eliminated and each user recovers their own message based on the rescaled received signal given by
\begin{equation}
	\bar{y}_k = \alpha y_k = s_k + \alpha v_k.
\end{equation}
Each user hence experiences an amplified noise with the effective variance
\begin{equation}\label{eq:noise}
	\begin{aligned}
		\bar{\sigma}^2 &= \ex(|\alpha v_k|^2)=\alpha^2\sigma^2.
	\end{aligned}
\end{equation}

\subsection{Non-bijective Modulation in CIP}
CIP aims to reduce the noise amplification effect of ZF by minimizing $\alpha^2$ in \eqref{eq:noise}. To facilitate the discussion on constellation design, we adopt the framework in~\cite{Liu2022}, which models CIP as the cascade of the channel-dependent ZF precoder \eqref{eq:zf} and a non-bijective modulation procedure. Specifically, for a given $m_k$, $s_k$ is not necessarily mapped to a fixed constellation point, but is selected from a message-dependent feasible region $\mathcal{R}(m_k)$. The modulation function of CIP is therefore relaxed from a bijective function of $\mathbf{m}$ to
\begin{equation}\label{eq:cip}
	\begin{aligned}
		f(\mathbf{m},\mathbf{H}) =\argmin_{\mathbf{s}\in \ar(\mathbf{m})}\alpha^2 =\argmin_{\mathbf{s}\in \ar(\mathbf{m})}\mathbf{s}^\tpc(\mathbf{H}\mathbf{H}^\tpc)^{-1}\mathbf{s},
	\end{aligned}
\end{equation}
where $\ar(\mathbf{m})$ denotes the Cartesian product $\mathcal{R}(m_1) \times \cdots \times \mathcal{R}(m_K)$. This shows that for given $\mathbf{m}$ and $\mathbf{H}$, the optimal value of $\alpha^2$ depends on the following set of feasible regions
\begin{equation}\label{eq:ar}
	\mathcal{D}=\{\mathcal{R}(m)\subset \mathbb{C} \mid m \in \mathcal{M}\}.
\end{equation}
Conventionally, $\mathcal{D}$ is taken as the set of the distance-preserving CI regions of a point-based constellation $\mathcal{C}$~\cite{Haqiqatnejad2018}. The CI region corresponding to the  constellation point $c_m$ is defined as the set of all points that maintain at least the same distance as $c_m$ to each of its decision boundaries. Mathematically, the $m$th CI region is given by
\begin{equation}\label{eq:dpci}
	\mathcal{R}(m)
	=
	\Bigl\{
	s \in \mathcal{V}_m
	\Big|
	\operatorname{d}(s,\mathcal{E}_{m,i})
	\ge
	\operatorname{d}(c_m,\mathcal{E}_{m,i}),
	\ \forall i\in\mathcal{I}_m
	\Bigr\},
\end{equation}
where $\mathcal{V}_m$ denotes the Voronoi cell of $c_m$, $\mathcal{E}_{m,i}$ denotes the Voronoi edge (i.e. the maximum-likelihood (ML) decision boundary) shared between $c_m$ and its neighbor $c_i$, $\mathcal{I}_m$ collects the indices of all such neighbors, and $\operatorname{d}(s, \mathcal{E}) = \min_{u \in \mathcal{E}} |s - u|$
denotes the Euclidean distance from a point to a set. To aid interpretation of the definition, Fig.~\ref{fig:dpci} illustrates the CI regions of $16$-QAM. The constellation points at the edges correspond to half-line regions except for the four corner points, which correspond to two-dimensional bounded regions, while the CI regions of the interior points remain singleton sets.

By construction, the CI regions guarantee that the SER of the relaxed symbols does not exceed that of $\mathcal{C}$ for any fixed $\bar{\sigma}^2$ under ML detection~\cite{Haqiqatnejad2018}, ensuring that minimizing $\alpha^2$ translates directly to SER improvement. The SER under CIP, denoted $P_{\mathrm{e}}$, therefore satisfies the following union bound for a general $M$-ary $\mathcal{C}$ (see e.g.,~\cite[Sec.~4.2]{Proakis2008})
\begin{equation}\label{eq:ub}
	P_{\mathrm{e}} \le (M-1)Q\bigg(\sqrt{\frac{d_\mathrm{min}^2}{2\bar{\sigma}^2}}\bigg) = (M-1)Q\bigg(\sqrt{\frac{d_\mathrm{min}^2}{2\alpha^2\sigma^2}}\bigg),
\end{equation}
where $d_{\min}$ denotes the minimum Euclidean distance of $\mathcal{C}$ and $Q(\cdot)$ denotes the tail distribution function of the standard normal distribution. Since the bound is monotonically decreasing in $\alpha^2$ for fixed $d_{\min}$, minimizing $\alpha^2$ directly reduces the SER upper bound. Consequently, for constellations sharing the same $d_{\min}$, $\alpha^2$ serves as a consistent proxy for SER performance under CIP.

\begin{figure}
	\centering
	\includegraphics[width=0.8\linewidth]{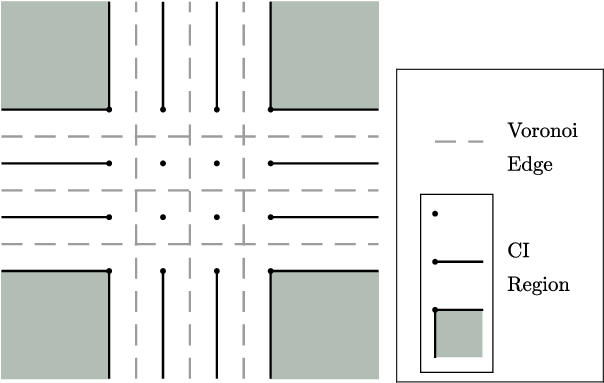}
	\caption{CI regions of 16-QAM. The dotted lines represent the Voronoi edges, and the shaded regions, rays, and singleton points denote the CI regions of corner, edge, and interior constellation points, respectively.}
	\label{fig:dpci}
\end{figure}

\section{Revisiting QAM-Based CIP}\label{sec3}
This section presents a novel analysis of the conventional QAM-based CIP 
problem, revealing the fundamental limitation imposed by the fixed-sign 
CI regions. These analytical insights directly motivate 
the novel constellation designs proposed in Section~\ref{sec4}.

Consider the square QAM constellation $\mathcal{C}_{\text{QAM}}$ with $M=L^2$, where $L\ge 4$ is an even integer. A symbol $s \in \mathcal{C}_{\text{QAM}}$ can be decomposed as
\begin{equation}
	s = \mathfrak{R}(s) + j\mathfrak{I}(s),
\end{equation}
where both the in-phase and quadrature components $\mathfrak{R}(s), \mathfrak{I}(s) \in \mathbb{R}$ take values in the same $L$-ary PAM constellation $\mathcal{C}_{\text{PAM}}$. In ascending order, the $\ell$th point in $\mathcal{C}_{\text{PAM}}$ is given by
\begin{equation}\label{eq:pam}
	c_\ell = \big(\ell - \frac{L-1}{2}\big)d_{\min}, \quad \ell \in \mathcal{L},
\end{equation}
where $\mathcal{L} = \{0,1,\cdots,L-1\}$ denotes the message set per real dimension. Without loss of generality, we set $d_{\min} = 2$ for notational simplicity throughout the remainder of this paper. \eqref{eq:pam} then simplifies to $c_\ell = 2\ell-L+1$. 

The CI regions of $L$-PAM derived from \eqref{eq:dpci} can be represented as 
\begin{equation}\label{eq:dpci_pam}
	\mathcal{R}_{\text{PAM}}(\ell) =
	\begin{cases}
		\{\,2\ell - L + 1\,\}, & \ell = 1,\dots,L-2, \\[2pt]
		(-\infty,-L+1], & \ell = 0, \\[2pt]
		[L-1,+\infty), & \ell = L-1 .
	\end{cases}
\end{equation}
Let $\phi:\mathbb{R}^2\!\to\!\mathbb{C}$ be the canonical identification $\phi(a,b)=a+j b$. Then the CI regions of $M$-QAM are given by
\begin{equation}
	\label{eq:CQAM_map}
	\mathcal{C}_{\text{QAM}}
	= \phi(\mathcal{C}_{\text{PAM}}\times\mathcal{C}_{\text{PAM}}).
\end{equation}
Owing to the independence of the real and imaginary dimensions in $\mathcal{C}_{\text{QAM}}$, it is convenient to analyze the optimization problem in \eqref{eq:cip} under QAM in the real domain. Using $\dot{(\cdot)}$ to denote the real-valued representation of a vector or matrix obtained by widely linear decomposition~\cite{Picinbono1995}, i.e., $\dot{\mathbf{s}}\in\mathbb{R}^{2K}$ and $\dot{\mathbf{H}}\in\mathbb{R}^{2K\times 2N_\mathrm{t}}$, and noting that $(\dot{\mathbf{H}}\dot{\mathbf{H}}^\tp)^{-1}$ is the real-valued representation of $(\mathbf{H}\mathbf{H}^\tpc)^{-1}$, \eqref{eq:cip} is rewritten as
\begin{subequations}\label{eq:cipqam}
	\begin{align}
		\arg\min_{\dot{\mathbf{s}}}~ &\dot{\mathbf{s}}^\mathsf{T}(\dot{\mathbf{H}}\dot{\mathbf{H}}^\mathsf{T})^{-1}\dot{\mathbf{s}} \\
		\mathrm{s.t.}~&
		\dot{\mathbf{s}}_{\mathcal{I}_\mathrm{in}}=2\boldsymbol{\ell}_{\mathcal{I}_\mathrm{in}}-L+1,\\
		& \dot{\mathbf{s}}_{\mathcal{I}_-}\le(-L+1)\mathbf{1}_{|\mathcal{I}_-|},\label{eq:cipqam_c}\\
		& \dot{\mathbf{s}}_{\mathcal{I}_+}\ge(L-1)\mathbf{1}_{|\mathcal{I}_+|},\label{eq:cipqam_d}
	\end{align}
\end{subequations}
where $\boldsymbol{\ell} = [(\mathbf{m} \bmod L)^\tp,~(\lfloor \mathbf{m}/L \rfloor)^\tp]^\tp \in \mathcal{L}^{2K}$ collects the real-domain messages corresponding to $\mathbf{m}$, with $\bmod$ denoting the elementwise modulo operator and $\lfloor\cdot\rfloor$ denoting the elementwise floor operator, the index sets $\mathcal{I}_{\mathrm{in}}$, $\mathcal{I}_{-}$, and $\mathcal{I}_{+}$ partition the entries of $\boldsymbol{\ell}$ according to the three cases in \eqref{eq:dpci_pam}, corresponding to interior, lower-end, and upper-end feasible regions, respectively, $(\cdot)_{\mathcal{I}}$ denotes the subvector formed by selecting the entries indexed by $\mathcal{I}$, and all vector inequalities are understood elementwise. 

According to \eqref{eq:dpci_pam}, an average number of $2K/L$ DoF in \eqref{eq:cipqam_c} and \eqref{eq:cipqam_d} can be exploited to minimize the objective. These DoF are constrained within feasible regions with fixed signs. To obtain analytical insight into the penalty introduced by these constraints, we analyze the relaxed problem of \eqref{eq:cipqam} which removes the constraints in \eqref{eq:cipqam_c} and \eqref{eq:cipqam_d}. Let $\mathcal{I}_\mathrm{end} = \mathcal{I}_-\cup\mathcal{I}_+$ denote the end-symbol index set. The relaxed problem is equivalent to real-domain ZF precoding~\cite{Zhang2017} based on $\dot{\mathbf{s}}_{\ids_\mathrm{in}}$ and ignoring the interference to $\dot{\mathbf{s}}_{\ids_\mathrm{end}}$. Therefore, the relaxed solution, denoted $\dot{\mathbf{s}}'$, satisfies
\begin{equation}\label{eq:xp}
	\dot{\mathbf{H}}^\dagger \dot{\mathbf{s}}' =\dot{\mathbf{H}}_{\mathcal{I}_\mathrm{in}}^\dagger \dot{\mathbf{s}}_{\mathcal{I}_\mathrm{in}},
\end{equation}
where $\dot{\mathbf{H}}_{\mathcal{I}_\mathrm{in}}$ denotes the submatrix constructed by the rows in $\dot{\mathbf{H}}$ with indices in $\mathcal{I}_\mathrm{in}$. Let $\alpha'^{2}=\dot{\mathbf{s}}'^{\mathsf{T}}(\dot{\mathbf{H}}\dot{\mathbf{H}}^\mathsf{T})^{-1}\dot{\mathbf{s}}'$ 
denote the optimal relaxed objective value. According to \eqref{eq:xp},
\begin{equation}\label{eq:ap}
	\begin{aligned}
		&\mathbb{E}(\alpha'^2)
		= \mathbb{E}\big(\dot{\mathbf{s}}_{\mathcal{I}_\mathrm{in}}^\mathsf{T}
		(\dot{\mathbf{H}}_{\mathcal{I}_\mathrm{in}}\dot{\mathbf{H}}_{\mathcal{I}_\mathrm{in}}^\mathsf{T})^{-1}
		\dot{\mathbf{s}}_{\mathcal{I}_\mathrm{in}}\big)\\
		=& E_{s,\mathrm{in}}\mathbb{E}\big(\mathrm{tr}\big(
		(\dot{\mathbf{H}}_{\mathcal{I}_\mathrm{in}}\dot{\mathbf{H}}_{\mathcal{I}_\mathrm{in}}^\mathsf{T})^{-1}
		\big)\big)\\ 
		=& E_{s,\mathrm{in}}\sum_{n=0}^{2K}\pr\{|\ids_\mathrm{in}|=n\}
		\mathbb{E}\big(\mathrm{tr}\big((\dot{\mathbf{H}}_{\mathcal{I}_\mathrm{in}}
		\dot{\mathbf{H}}_{\mathcal{I}_\mathrm{in}}^\mathsf{T})^{-1}\big)\big)
		\big|_{|\mathcal{I}_\mathrm{in}|=n},
	\end{aligned}
\end{equation}
where $ E_{s,\mathrm{in}} = \ex(\dot{s}_i^2)=\sum_{\ell=1}^{L-2}(2\ell-L+1)^2/(L-2),~i\in\mathcal{I}_\mathrm{in}$ denotes the average per-real-dimension energy of interior symbols. The second equality holds since $\dot{\mathbf{s}}_{\mathcal{I}_\mathrm{in}}$ 
and $\dot{\mathbf{H}}_{\mathcal{I}_\mathrm{in}}$ are independent.

\begin{proposition}\label{prop1}
	Let $\eta = N_\mathrm{t}/K$ denote the antenna-to-user ratio. 
	The expected optimal relaxed objective satisfies
	\begin{equation}\label{eq:ap_lb}
		\mathbb{E}(\alpha'^{2}) \ge 
		\frac{2E_{s,\mathrm{in}}}{\eta L/(L-2)-1}.
	\end{equation}
	The inequality holds exactly when $\eta > 1$, and 
	asymptotically as $K \to \infty$ when $\eta = 1$.
\end{proposition}

\begin{proof}
	See \textsc{Appendix~\ref{apl1}}.
\end{proof}
Consider a representative scenario with $16$-QAM (i.e., $L=4$) and a 
fully-loaded system (i.e., $\eta=1$). Substituting these values into~\eqref{eq:ap_lb} 
gives $\mathbb{E}(\alpha'^{2})\ge 2$, indicating that relaxing the 
end symbol constraints can at best achieve near-AWGN performance. 
However, the actual $16$-QAM CIP performance falls far short of 
this limit, as the results in~\cite{Li2020interference} show only $10$~dB SNR gain over ZF precoding at a target bit error rate (BER) 
of $10^{-3}$, whereas the AWGN bound yields more than $20$~dB under the 
same setting. The following analysis reveals how the fixed-sign CI regions fundamentally limit the performance 
of QAM-based CIP.

Let $\mathbf{z} = \mathrm{sgn}(\dot{\mathbf{s}}_{\mathcal{I}_\mathrm{end}})$ 
be the sign pattern induced by the constraints on the end symbols. Among all $\mathbf{z}\in\{\pm1\}^{|\mathcal{I}_\mathrm{end}|}$, 
$\mathbf{z}' = \mathrm{sgn}(\dot{\mathbf{s}}'_{\mathcal{I}_\mathrm{end}})$ 
yields the feasible solution closest to the unconstrained 
minimizer $\dot{\mathbf{s}}'$ in Euclidean distance, and therefore 
tends to achieve an objective value close to $\alpha'^2$. However, 
the following result shows that $\mathbf{z}$ exhibits systematic 
misalignment with $\mathbf{z}'$.
\begin{proposition}\label{prop2}
	Let $z'_i$ and $z_i$ denote the $i$th entries of 
	$\mathbf{z}'$ and $\mathbf{z}$, respectively, 
	\begin{equation}
		\Pr(z'_i = z_i) = \tfrac{1}{2}, \quad 
		\forall\, i \in \{1,\cdots,|\mathcal{I}_{\mathrm{end}}|\}.
	\end{equation}

\end{proposition}
\begin{proof}
	See \textsc{Appendix}~\ref{apl2}.
\end{proof}
Consequently, $\mathbf{z}$ agrees with $\mathbf{z}'$ for only half 
of the end symbol dimensions on average, contributing to the degradation in the objective relative to 
$\alpha'^2$. Next, we present a constraint modification which improves sign alignment.

\begin{proposition}\label{prop3}
	For \eqref{eq:cipqam}, arbitrarily partition $\mathcal{I}_+ = \mathcal{I}_{+,1}\cup\mathcal{I}_{+,2}$ 
	and $\mathcal{I}_- = \mathcal{I}_{-,1}\cup\mathcal{I}_{-,2}$ 
	with $|\mathcal{I}_{+,1}|=\lfloor|\mathcal{I}_+|/2\rceil$ and 
	$|\mathcal{I}_{-,1}|=\lfloor|\mathcal{I}_-|/2\rceil$, with 
	$\lfloor\cdot\rceil$ denoting rounding to the nearest integer, 
	and let $\mathcal{I}_{\mathrm{end},2} = \mathcal{I}_{+,2}\cup\mathcal{I}_{-,2}$. 
	Replace \eqref{eq:cipqam_c} and \eqref{eq:cipqam_d} with the 
	following constraints:
	\begin{subequations}\label{eq:mod}
		\begin{align}
			& \dot{\mathbf{s}}_{\mathcal{I}_{-,1}}=(-L+1)\mathbf{1}_{|\mathcal{I}_{-,1}|},\label{eq:moda}\\
			& \dot{\mathbf{s}}_{\mathcal{I}_{+,1}}=(L-1)\mathbf{1}_{|\mathcal{I}_{+,1}|},\label{eq:modb}\\
			& |\dot{\mathbf{s}}_{\mathcal{I}_{\mathrm{end},2}}|\ge(L-1)\mathbf{1}_{|\mathcal{I}_{\mathrm{end},2}|},\label{eq:sf}
		\end{align}
	\end{subequations}
	where $|\dot{\mathbf{s}}_{\mathcal{I}_{\mathrm{end},2}}|$ is understood elementwise. Under this 
	modification, $3/4$ of the entries in $\mathbf{z}$ can be 
	set equal to their counterparts in $\mathbf{z}'$ on average.
\end{proposition}
\begin{proof}
	Let $\mathcal{I}_{\mathrm{end},1} = \mathcal{I}_{+,1}\cup\mathcal{I}_{-,1}$. 
	$\dot{\mathbf{s}}_{\mathcal{I}_{\mathrm{end},1}}$ and 
	$\dot{\mathbf{s}}_{\mathcal{I}_{\mathrm{end},2}}$ each contain
	approximately half of the entries in $\dot{\mathbf{s}}_{\mathcal{I}_\mathrm{end}}$. 
	By Proposition~\ref{prop2}, each entry of 
	$\dot{\mathbf{s}}_{\mathcal{I}_{\mathrm{end},1}}$ aligns in sign 
	with its counterpart in $\dot{\mathbf{s}}'_{\mathcal{I}_{\mathrm{end},1}}$ 
	with probability $1/2$, contributing $1/2\cdot1/2=1/4$ aligned 
	entries on average. On the other hand, 
	$\dot{\mathbf{s}}_{\mathcal{I}_{\mathrm{end},2}}$ can always be 
	set to align with $\dot{\mathbf{s}}'_{\mathcal{I}_{\mathrm{end},2}}$ 
	due to the absolute operator in \eqref{eq:sf}, contributing 
	another $1/2$ aligned entries. Therefore, 
	on average $1/4+1/2=3/4$ of entries in $\mathbf{z}$ can be aligned with their counterparts in $\mathbf{z}'$.
\end{proof}
Consequently, the modification achieves an average increase of 
$1/4$ in the proportion of sign alignment. Moreover, it guarantees at least 
$|\mathcal{I}_{\mathrm{end},2}|$ aligned signs, preventing highly 
unaligned cases where most signs in $\mathbf{z}$ disagree with 
$\mathbf{z}'$ and consequently lead to a large objective value.

Despite the advantage in sign alignment, the sign flexibility in \eqref{eq:sf} introduces ambiguity that prevents 
users from distinguishing between messages $\ell=0$ and $\ell=L-1$ 
according to \eqref{eq:dpci_pam}, violating the SER upper bound 
in~\eqref{eq:ub}. Therefore, this modification is not directly 
applicable to practical CIP. Nevertheless, the next section 
demonstrates that feasible constellation schemes can be developed based on this modification.

\section{Region-Based Constellation Designs for CIP}\label{sec4}
This section proposes two novel constellation schemes named ME-QAM and RM-QAM which involve non-convex feasible regions. Such a region cannot be modeled as the CI region of a constellation point, which is convex by the definition in \eqref{eq:dpci}. As a result, the proposed schemes cannot be described by the conventional point-based constellation model, which motivates the following RBC model.

\begin{definition}[RBC]
	An $M$-ary RBC is a collection of $M$ feasible regions indexed by 
	distinct messages, given by
	\begin{equation}
		\mathcal{D} = \{\mathcal{R}(m)\subset \mathbb{C} \mid m \in \mathcal{M}\}.
	\end{equation}
\end{definition}
This model allows $\mathcal{R}(m)$ to be any subset of $\mathbb{C}$, removing the convexity restriction inherent in the CI region construction. The optimal RBC for the CIP modulation function \eqref{eq:cip} with given $K$, $N_\mathrm{t}$, and distribution of $\mathbf{H}$ is then obtained by solving
\begin{equation}\label{eq03}
	\argmin_{\mathcal{D}\in\mathcal{A}}~\mathbb{E}(\alpha^2),
\end{equation}
where the expectation is taken over both $\mathbf{m}$ and $\mathbf{H}$, and $\mathcal{A}$ denotes the class of RBCs that satisfy the SER upper bound in \eqref{eq:ub}. Due to the removal of the reference constellation, $d_{\min}$ of a RBC is redefined as the minimum distance between two feasible regions, which is given by
\begin{equation}
	d_{\min} = \min_{m \ne n,\, p \in \mathcal{R}(m),\, q \in \mathcal{R}(n)} |p - q|.
\end{equation}

Directly solving~\eqref{eq03} is analytically intractable, as the feasible 
set $\mathcal{A}$ admits no tractable general parameterization. Rather than 
deriving the globally optimal RBC, our goal is to design feasible RBCs which inherit the sign-alignment advantage obtained by the constraint modification in \eqref{eq:mod}.

\subsection{ME-QAM}\label{sec4b}
\begin{definition}[ME-QAM]
	The $M$-ary ME-QAM RBC with $M=L^2$ is defined by
	\begin{equation}\label{eq:meqam}
		\mathcal{D}_{\mathrm{ME\text{-}QAM}}
		=
		\phi(\mathcal{D}_{\mathrm{ME\text{-}PAM}} \times \mathcal{D}_{\mathrm{ME\text{-}PAM}}),
	\end{equation}
	where $\mathcal{D}_{\mathrm{ME\text{-}PAM}}$ is the one-dimensional $L$-ary ME-PAM RBC whose feasible regions are defined by
	\begin{equation}\label{eq:mepam}
		\mathcal{R}(\ell)
		=
		\begin{cases}
			\{\,2\ell-L+2\,\}, & \ell = 0,1,\dots,L-2,\\[2pt]
			(-\infty,-L]\cup[L,\infty), & \ell = L-1,
		\end{cases}
	\end{equation}
	where $\mathcal{R}(L-1)$ is referred to as the \emph{sign-flexible (SF)} region.
\end{definition}
\textsc{Appendix~\ref{apme}} proves that ME-QAM satisfies the SER upper bound in \eqref{eq:ub} and hence belongs to class $\mathcal{A}$. 

Fig.~\ref{fig:me} illustrates the ME-QAM RBC for the representative case $M=16$, where regions sharing the same label are assigned to the same message $m$. In particular, the four corner regions labeled $0$ correspond to the symbol assigned SF regions in both dimensions. 

The real-domain CIP optimization problem with ME-QAM is given by
\begin{equation}\label{eq:cipme}
	\begin{aligned}
		\arg\min_{\dot{\mathbf{s}},\boldsymbol{\psi}}~ &\dot{\mathbf{s}}^\mathsf{T}(\dot{\mathbf{H}}\dot{\mathbf{H}}^\mathsf{T})^{-1}\dot{\mathbf{s}} \\
		\mathrm{s.t.}~&
		\dot{\mathbf{s}}_{\mathcal{I}_\mathrm{in,ME}}=2\boldsymbol{\ell}_{\mathcal{I}_\mathrm{in,ME}}-L+2,\\
		& \boldsymbol{\psi}\odot\dot{\mathbf{s}}_{\mathcal{I}_\mathrm{end,ME}}\ge L\mathbf{1}_{|\mathcal{I}_\mathrm{end,ME}|},\\
		& \boldsymbol{\psi}\in\{\pm1\}^{|\mathcal{I}_\mathrm{end,ME}|},
	\end{aligned}
\end{equation}
where $\mathcal{I}_\mathrm{in,ME}$ and $\mathcal{I}_\mathrm{end,ME}$ partition the entries of $\boldsymbol{\ell}$ according to the first and second cases in \eqref{eq:mepam}, respectively, and $\odot$ denotes the Hadamard product.

Next, we demonstrate the similarity between problems \eqref{eq:cipme} and \eqref{eq:cipqam} with modifications in \eqref{eq:mod}. Since $\ex(|\mathcal{I}_\mathrm{in,ME}|)=\ex(|\mathcal{I}_\mathrm{in}|+|\mathcal{I}_{-,1}|+|\mathcal{I}_{+,1}|)=2K(L-1)/L$, both problems involve the same average number of singleton symbols which are symmetrically distributed within the range $[-L+1,L-1]$ or $[-L+2,L-2]$, respectively. On the other hand, since $\ex(|\mathcal{I}_\mathrm{end,ME}|)=\ex(|\mathcal{I}_\mathrm{end,2}|)=2K/L$, \eqref{eq:cipme} has the same amount of SF constraints on average as in \eqref{eq:sf}, whose boundaries are slightly shifted outward by $1$ unit on the real axis. Therefore, both problems exhibit highly similar structure, indicating that ME-QAM inherits the sign-alignment benefits described in \textit{Proposition}~\ref{prop3}.

\begin{figure}[t]
	\centering
	\begin{subfigure}[t]{0.8\linewidth}
		\centering
		\includegraphics[width=\linewidth]{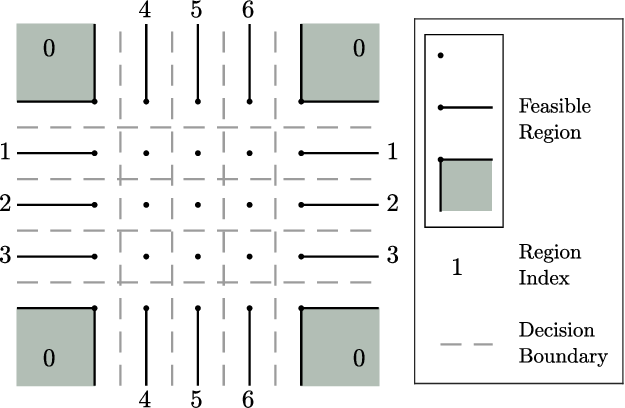}
		\caption{ME-QAM}
		\label{fig:me}
		\vspace{0.2cm}
	\end{subfigure}
	\begin{subfigure}[t]{0.8\linewidth}
		\centering
		\includegraphics[width=\linewidth]{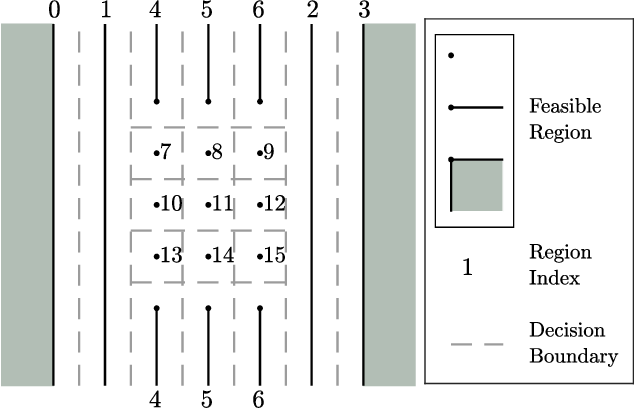}
		\caption{RM-QAM}
		\label{fig:rm}
	\end{subfigure}
	\caption{
		Feasible regions of (a) ME-QAM and (b) RM-QAM for $M=16$. Regions sharing the same label are assigned to the same message $m$. The feasible regions with labels $\{0,1,2,3\}$ in (a) ME-QAM are replaced by the corresponding regions in (b) RM-QAM, which are unbounded in the imaginary direction. The rest of the feasible regions remain identical in both constellations.}
	\label{fig:me_rm}
\end{figure}

Despite the advantage, ME-QAM incurs an energy penalty 
relative to QAM due to its outward-shifted constellation, which increases 
the baseline value of $\alpha^2$. For conventional point-based constellations, 
this penalty is naturally quantified by the average symbol energy. We extend 
this measure to RBCs by adopting the average minimum symbol energy given by
\begin{equation}
	E_s = \frac{1}{M}\sum_{m=1}^{M}\min_{s_m\in\mathcal{R}(m)}|s_m|^2,
\end{equation}
which selects the lowest-energy point in each feasible region and 
is independent of the channel distribution. For ME-QAM, $E_s = {2(M+2)}/{3}$, which is derived from 
\eqref{eq:meqam} and \eqref{eq:mepam}. Compared with the average symbol 
energy of QAM, ${2(M-1)}/{3}$, ME-QAM increases $E_s$ by a 
factor of $(M+2)/(M-1)$, which tends to $1$ as $M\rightarrow\infty$, 
confirming that the relative energy penalty diminishes for large constellation sizes.

\subsection{RM-QAM}
In this subsection, we propose the RM-QAM RBC, which is developed from 
ME-QAM by replacing a portion of the SF regions with further relaxed regions that are 
unconstrained in the imaginary dimension.
\begin{definition}[RM-QAM]
	The $M$-ary RM-QAM RBC with $M=L^2$ is defined by
	\begin{equation}\label{eq:rm}
		\begin{aligned}
			&\mathcal{D}_{\mathrm{RM\text{-}QAM}}
			=
			\left\{
			\mathcal{R}(m)\in\mathcal{D}_{\mathrm{ME\text{-}QAM}}
			:\;
			\min\bigl(|\mathfrak{R}(\mathcal{R}(m))|\bigr)\neq L
			\right\} \\
			&\cup
			\{
			\mathcal{R} :
			\mathfrak{R}(\mathcal{R})=\{\pm(L+2\ell)\},\;\ell=0,\dots,L/2-2,\\
			&~~\quad\mathfrak{I}(\mathcal{R})=\mathbb{R}
			\} \\
			&\cup
			\left\{
			\mathcal{R} :
			\mathfrak{R}(\mathcal{R})=[2L-2,\infty),\;
			\mathfrak{I}(\mathcal{R})=\mathbb{R}
			\right\} \\
			&\cup
			\left\{
			\mathcal{R} :
			\mathfrak{R}(\mathcal{R})=(-\infty,-2L+2],\;
			\mathfrak{I}(\mathcal{R})=\mathbb{R}
			\right\}.
		\end{aligned}
	\end{equation}
\end{definition}
Fig.~\ref{fig:rm} illustrates the $16$-ary RM-QAM RBC, where the labels 
represent the messages $m$ assigned to each region. The following discussion 
is based on this example; the extension to other constellation sizes is straightforward. The first set in 
\eqref{eq:rm} retains all feasible regions from ME-QAM whose real 
components do not span both ends of the RBC, corresponding to the singleton 
regions $m\in\mathcal{M}_1=\{7,\cdots,15\}$ and the SF regions $m\in\mathcal{M}_2=\{4,5,6\}$ 
in Fig.~\ref{fig:rm}. The remaining three sets replace the SF 
regions of ME-QAM with new regions that are unconstrained in the imaginary 
dimension: the second set introduces regions with fixed real components 
($m\in\mathcal{M}_3=\{1,2\}$), while the third and fourth sets 
introduce regions with semi-infinite real components at the right and left 
ends of the RBC ($m\in\mathcal{M}_4=\{0,3\}$), respectively. 
 
\textsc{Appendix~\ref{aprm}} proves that RM-QAM belongs to 
class $\mathcal{A}$. 

The real-domain CIP optimization problem with RM-QAM is given by
\begin{equation}\label{eq:ciprm}
	\begin{aligned}
		\arg\min_{\dot{\mathbf{s}},\boldsymbol{\psi}}~& \dot{\mathbf{s}}^\mathsf{T}(\dot{\mathbf{H}}\dot{\mathbf{H}}^\mathsf{T})^{-1}\dot{\mathbf{s}} \\
		\mathrm{s.t.}~&
		\dot{\mathbf{s}}_{\ids_1}=\mathfrak{R}(\ar(\mathbf{m}_{\ids_1})),~\dot{\mathbf{s}}_{\ids_1+K}=\mathfrak{I}(\ar(\mathbf{m}_{\ids_1})),\\
		&\dot{\mathbf{s}}_{\ids_2}=\mathfrak{R}(\ar(\mathbf{m}_{\ids_2})),~\boldsymbol{\psi}\odot\dot{\mathbf{s}}_{\ids_2+K}\ge L\mathbf{1}_{|\ids_2|},\\
		&\dot{\mathbf{s}}_{\ids_3}=\mathfrak{R}(\ar(\mathbf{m}_{\ids_3})), \boldsymbol{\theta}\odot\dot{\mathbf{s}}_{\ids_4}\ge (2L-2)\mathbf{1}_{|\ids_4|},\\
		&\boldsymbol{\psi}\in\{\pm1\}^{|\mathcal{I}_{2}|},~\dot{\mathbf{s}}_{(\ids_3\cup\ids_4)+K}\in\mathbb{R}^{|\ids_3\cup\ids_4|},
	\end{aligned}
\end{equation}
where $\ids_n$ denotes the set of indices of $m_k\in\mathcal{M}_n$ in $\mathbf{m}$ and $\boldsymbol{\theta}$ denotes the fixed sign pattern corresponding to $\ids_4$. Since the first and second halves of $\dot{\mathbf{s}}$ correspond to the 
real and imaginary parts of $\mathbf{s}$, respectively, the indices shifted by $K$ in \eqref{eq:ciprm} correspond to imaginary-part entries, while the unshifted indices correspond to the real-part entries.

Due to the asymmetric structure, RM-QAM further increases $E_s$ relative 
to ME-QAM, e.g., by approximately $0.3$~dB for $M=16$. In contrast, the free DoF in \eqref{eq:ciprm} is beneficial for reducing $\alpha^2$. As long as the gain brought by the free DoF in \eqref{eq:ciprm} outweighs the penalty in $E_s$, RM-QAM can outperform ME-QAM. This gain can be approximated by $\Delta$ in the following proposition.

\begin{proposition}\label{prop4}
	Arbitrarily partition $\ids_\mathrm{end,ME}$ in \eqref{eq:cipme} into $ \mathcal{I}_{A}\cup\mathcal{I}_{B}$ such that $|\mathcal{I}_{B}| = \lfloor |\mathcal{I}_\mathrm{end}|/2$$\rceil$. Define $\mathbf{Q}=(\dot{\mathbf{H}}\dot{\mathbf{H}}^\mathsf{T})^{-1}$, and let $\dot{\mathbf{s}}^\star$ be an optimal solution. Removing the constraints on $\dot{\mathbf{s}}_{\mathcal{I}_B}$ results in a reduction in the optimal objective value. Denote this reduction as $\Delta$, which satisfies the following lower bound
	\begin{equation}
		\Delta\ge
		\underline{\Delta}
		=
		\mathbf{g}_{\mathcal{I}_B}^\mathsf{T}
		\mathbf{Q}_{\mathcal{I}_B\mathcal{I}_B}^{-1}
		\mathbf{g}_{\mathcal{I}_B},
	\end{equation}
	where $\mathbf{g}=\mathbf{Q}\dot{\mathbf{s}}^\star$, and $\mathbf{A}_{\mathcal{I}_{a}\mathcal{I}_{b}}$ denotes the submatrix of $\mathbf{A}$ formed by selecting the rows and columns indexed by $\mathcal{I}_{a}$ and $\mathcal{I}_{b}$.
\end{proposition}
\begin{proof}
	See \textsc{Appendix~\ref{apl3}}.
\end{proof}
Since the characteristics of $\mathbf{Q}_{\mathcal{I}_B\mathcal{I}_B}^{-1}$ depend on both $M$ and the dimension of $\mathbf{H}$, the performance gain of RM-QAM varies in different scenarios, which will be further examined via simulations in Section~\ref{sec5}. 

\begin{algorithm}[t]
	\caption{Predicted-Sign QP (PS-QP) for ME-QAM and RM-QAM.}
	\label{alg:sps}
	\begin{algorithmic}[1]
		\Statex Input: $\mathbf{H}$, $\boldsymbol{\ell}$, $M$, \textbf{Output}: $\dot{\mathbf{s}}^\diamond$.
		\State For ME-QAM, obtain $\mathcal{I}_\mathrm{SF} = \mathcal{I}_\mathrm{end,ME}$ and 
		$\mathcal{I}_\mathrm{F} = \mathcal{I}_\mathrm{in,ME}$ according to \eqref{eq:mepam} and \eqref{eq:cipme}. For RM-QAM, $\mathcal{I}_\mathrm{SF} = 
		\mathcal{I}_2 + K$ and  $\mathcal{I}_\mathrm{F} = 
		(\mathcal{I}_1+K)\cup\bigcup_{n=1}^4\mathcal{I}_n$  according to \eqref{eq:rm} and \eqref{eq:ciprm}.
		\State Compute $\hat{\boldsymbol{\psi}}$ according to \eqref{eq:sgn_pre}.
		\State Substitute $\boldsymbol{\psi} \leftarrow \hat{\boldsymbol{\psi}}$ 
		in \eqref{eq:cipme}/\eqref{eq:ciprm} and solve the resulting QP to obtain $\dot{\mathbf{s}}^\diamond$.
	\end{algorithmic}
\end{algorithm}
\subsection{Algorithms for CIP with ME-QAM and RM-QAM}
The non-convex MIQPs \eqref{eq:cipme} and \eqref{eq:ciprm} can be 
solved to global optimality by enumerating all feasible $\boldsymbol{\psi}$, each 
yielding an LCQP that can be solved efficiently by optimization toolboxes such as CVX. This procedure is referred to as the \emph{full-search QP (FS-QP)} algorithm. Although optimal, FS-QP incurs 
an exponential complexity substantially exceeding that of 
QAM-based CIP. Motivated by the sign alignment analysis in Section~\ref{sec3}, we propose the following \emph{predicted-sign QP (PS-QP)} algorithm, which generates the predicted sign pattern with a closed-form expression.

Let $\mathcal{I}_\mathrm{SF}$ and $\mathcal{I}_\mathrm{F}$ denote 
the index sets of entries in $\dot{\mathbf{s}}$ with SF and fixed-sign regions, respectively. For 
ME-QAM, $\mathcal{I}_\mathrm{SF} = \mathcal{I}_\mathrm{end,ME}$ and 
$\mathcal{I}_\mathrm{F} = \mathcal{I}_\mathrm{in,ME}$ in 
\eqref{eq:cipme}. For RM-QAM, $\mathcal{I}_\mathrm{SF} = 
\mathcal{I}_2 + K$ and $\mathcal{I}_\mathrm{F} = 
(\mathcal{I}_1+K)\cup\bigcup_{n=1}^4\mathcal{I}_n$ 
in \eqref{eq:ciprm}. PS-QP first generates a predicted sign pattern as
\begin{equation}\label{eq:sgn_pre}
	\hat{\boldsymbol{\psi}}=\operatorname{sgn}
	(\dot{\mathbf{H}}_{\mathcal{I}_\mathrm{SF}}
	\dot{\mathbf{H}}^\dagger_{\mathcal{I}_\mathrm{F}}
	\dot{\mathbf{s}}_{\mathcal{I}_\mathrm{F}}).
\end{equation}
Particularly for RM-QAM, each entry in $\dot{\mathbf{s}}_{\mathcal{I}_4}$ is 
fixed to the boundary of its feasible region
when performing sign predictions. $\boldsymbol{\psi}$ is then 
substituted by $\hat{\boldsymbol{\psi}}$ in \eqref{eq:cipme} (or 
\eqref{eq:ciprm}), and the resulting LCQP is solved to obtain 
the suboptimal solution $\dot{\mathbf{s}}^\diamond$. The procedure of obtaining $\dot{\mathbf{s}}^\diamond$ is summarized in \textbf{Algorithm~\ref{alg:sps}}.

\begin{figure}[t]
	\begin{subfigure}[t]{0.47\linewidth}
		\centering
		\includegraphics[width=\linewidth]{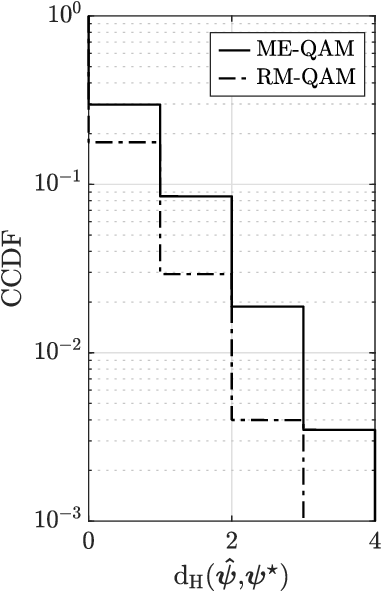}
		\caption{}
		\label{fig:hamdist}
	\end{subfigure}
	\begin{subfigure}[t]{0.47\linewidth}
		\centering
		\includegraphics[width=\linewidth]{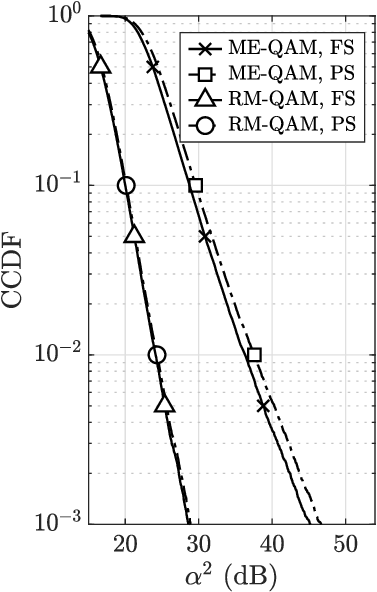}
		\caption{}
		\label{fig:algccdf}
	\end{subfigure}
	\caption{Comparison between PS-QP and FS-QP with $64$-ary ME-QAM and $16$-ary RM-QAM under $16\times 16$ MIMO. (a) CCDF of the Hamming distance between the optimal sign pattern $\boldsymbol{\psi}^\star$ obtained by FS-QP and the predicted sign pattern $\hat{\boldsymbol{\psi}}$. (b) CCDF of $\alpha^2$ obtained by FS-QP and PS-QP.}
	\label{fig:alg}
\end{figure}

The per-symbol-duration computational complexity of PS-QP consists of two 
parts: the sign prediction step and the LCQP solution step. For the sign 
prediction, the required matrix inversion $(\dot{\mathbf{H}}_{\mathcal{I}_\mathrm{F}}
\dot{\mathbf{H}}_{\mathcal{I}_\mathrm{F}}^T)^{-1}$ can be efficiently obtained 
from the precomputed $\mathbf{Q}=(\dot{\mathbf{H}}\dot{\mathbf{H}}^T)^{-1}$ 
via the block matrix inversion identity~\cite{thematrix}:
\begin{equation}
	\dot{\mathbf{H}}^\dagger_{\mathcal{I}_\mathrm{F}} = \dot{\mathbf{H}}^\mathsf{T}_{\mathcal{I}_\mathrm{F}}(\mathbf{Q}_{\mathcal{I}_\mathrm{F}\mathcal{I}_\mathrm{F}} - \mathbf{Q}_{\mathcal{I}_\mathrm{F}\mathcal{I}^\mathsf{c}_\mathrm{F}} \mathbf{Q}_{\mathcal{I}^\mathsf{c}_\mathrm{F}\mathcal{I}^\mathsf{c}_\mathrm{F}}^{-1} \mathbf{Q}_{\mathcal{I}^\mathsf{c}_\mathrm{F}\mathcal{I}_\mathrm{F}})
\end{equation}
where the only per-symbol matrix inversion is $\mathbf{Q}_{\mathcal{I}^\mathsf{c}_\mathrm{F}\mathcal{I}^\mathsf{c}_\mathrm{F}}^{-1}$ of dimension $|\mathcal{I}^\mathsf{c}_\mathrm{F}|=\mathcal{O}(K/L)$, 
costing $\mathcal{O}((K/L)^3)$, while the subsequent matrix-matrix multiplications 
cost $\mathcal{O}(K^3/L)$, which dominates the sign prediction step. The LCQP solve 
using a standard interior-point method costs $\mathcal{O}(\bar{N}^{3.5})$ 
\cite{Boyd2004}, where $\bar{N}$ denotes the average number of 
entries in $\dot{\mathbf{s}}$ not fixed to singletons. For ME-QAM, RM-QAM, 
and QAM, we have $\bar{N}={2K}/{L}$, ${(2L+1)K}/{L^2}$, 
and ${4K}/{L}$, respectively, all of which scale as $\mathcal{O}(K/L)$ 
for fixed $L$, giving the same asymptotic QP complexity order 
$\mathcal{O}((K/L)^{3.5})$. Therefore, the overall per-symbol complexity of PS-QP is 
$\mathcal{O}(K^3/L + (K/L)^{3.5})$, which is significantly lower than  $\mathcal{O}(2^{|\mathcal{I}_\mathrm{SF}|}(K/L)^{3.5})$ 
of FS-QP, and comparable to that of QAM-based CIP at 
$\mathcal{O}((K/L)^{3.5})$.

\section{Simulation Results and Discussions}\label{sec5}
The objectives of the simulations are \begin{enumerate*}
	\item to validate the suboptimality of PS-QP against FS-QP;
	\item to demonstrate the reduction in $\alpha^2$ achieved by the proposed RBC schemes;
	\item to demonstrate the advantage of the proposed schemes in SER and examine the cases under imperfect CSI and channel coding.
\end{enumerate*}

\subsection{System Setup and Baselines}
All simulations are performed under the MU-MIMO system described in Section~\ref{sec2}. 

Unless otherwise specified, PS-QP is used for the CIP problem for both ME-QAM \eqref{eq:cipme} and RM-QAM \eqref{eq:ciprm}.

For comparison, three baseline schemes are considered: QAM-based CIP \eqref{eq:cipqam}, PSK-based CIP~\cite{Li2018}, and linear ZF precoding with QAM, where the latter serves as a linear precoding benchmark. Each baseline will be labeled as QAM, PSK, and ZF in the figures, respectively.

\begin{figure}[t]
	\centering
	\includegraphics[width=0.44\textwidth]{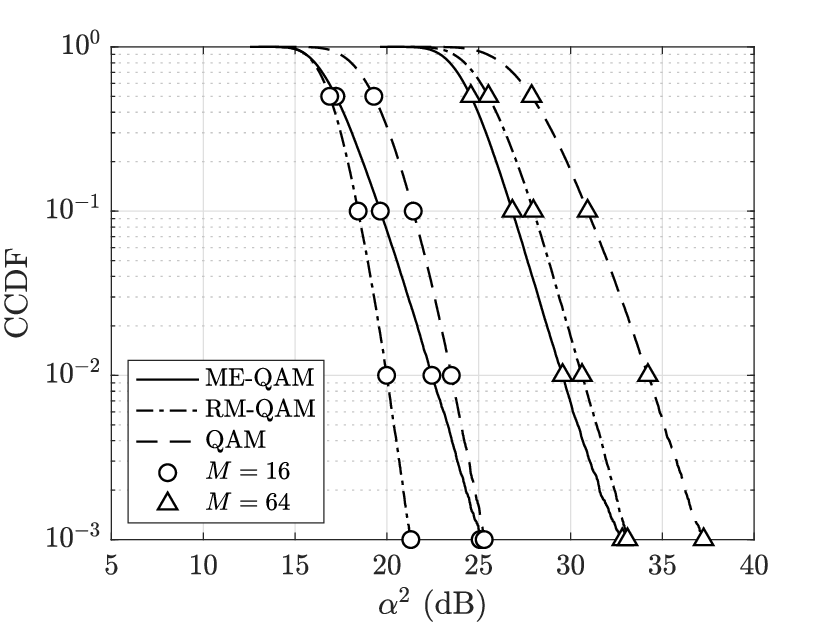}
	\caption{CCDF of $\alpha^2$ for ME-QAM, RM-QAM, and QAM with $M \in \{16, 64\}$, evaluated under $64\times64$ MIMO. A leftward shift of the CCDF curve indicates a stochastic reduction in $\alpha^2$ and improved CIP performance. }
	\label{fig:ccdf}
\end{figure}

\subsection{Experiment 1: PS-QP vs. FS-QP}
This experiment validates the sign prediction accuracy of PS-QP and its performance gap from FS-QP with $64$-ary ME-QAM and $16$-ary RM-QAM under $16\times 16$ MIMO. Fig.~\ref{fig:hamdist} shows the CCDF of the Hamming distance $\mathrm{d}_\mathrm{H}(\hat{\boldsymbol{\psi}},\boldsymbol{\psi}^\star)$, i.e., the number of sign disagreements between the predicted and optimal sign patterns. The prediction is exact (i.e., $\mathrm{d}_\mathrm{H}=0$) in approximately $80\%$ and $70\%$ of channel realizations for RM-QAM and ME-QAM, respectively, and the probability of $\mathrm{d}_\mathrm{H}$ exceeding $3$ is under $1\%$ for both. These results demonstrate that PS-QP predicts the optimal sign pattern with high accuracy. Fig.~\ref{fig:algccdf} shows the CCDF of $\alpha^2$ for FS-QP and PS-QP. For RM-QAM, the two curves are indistinguishable, while ME-QAM incurs a gap of approximately $1$\,dB, confirming that PS-QP achieves near-optimal performance at a substantially reduced complexity.

\subsection{Experiment 2: CCDF of $\alpha^2$}
This experiment demonstrates the reduction in $\alpha^2$ of the proposed schemes compared to QAM-based CIP. Fig.~\ref{fig:ccdf} 
compares the CCDF of $\alpha^2$ for the proposed schemes against QAM 
under $64\times 64$ MIMO with $M\in\{16,64\}$.

For $M=64$, both ME-QAM and RM-QAM exhibit a clear leftward shift 
of the CCDF relative to QAM across the entire evaluated range, 
corresponding to reductions in $\alpha^2$ of approximately $5$~dB 
and $4$~dB, respectively. The observed stochastic dominance confirms 
both a smaller average $\alpha^2$ and a reduced probability of 
high-$\alpha^2$ events.

For $M=16$, RM-QAM maintains a consistent reduction of over $3$~dB 
across the evaluated CCDF range. The improvement of ME-QAM, however, 
is less consistent, with its CCDF curve intersecting that of QAM 
near the $10^{-3}$ level. This degradation is attributable to two 
factors that are both more pronounced at smaller $M$: the $E_s$ 
penalty relative to QAM, and the outward shift of the SF region 
boundaries in ME-PAM by a factor of $L/(L-1)$ 
according to \eqref{eq:mepam} and \eqref{eq:pam}, which removes lower-energy candidates from 
the feasible set and increases the baseline $\alpha^2$.

\subsection{Experiment 3: SER under perfect CSI}

\begin{figure}[t]
	\centering
	\begin{subfigure}[t]{0.8\linewidth}
		\centering
		\includegraphics[width=\linewidth]{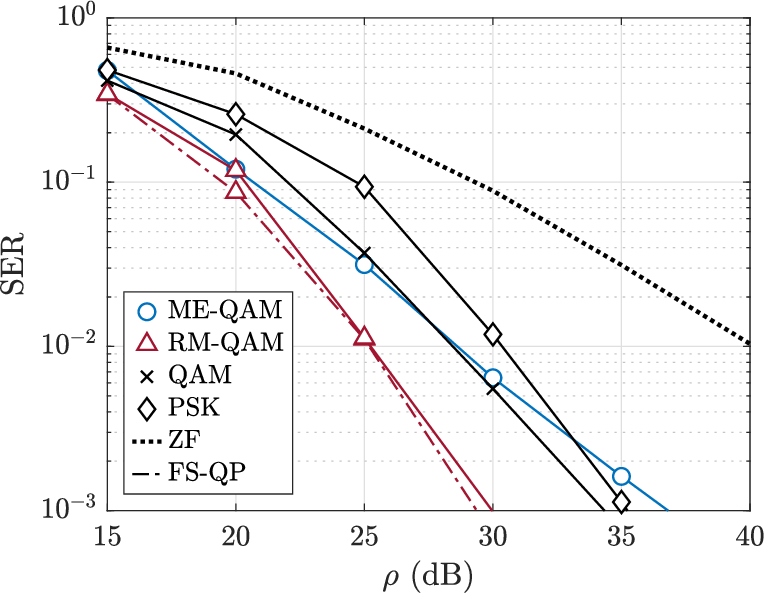}
		\caption{$M=16$}
		\label{fig:ser_k16_m16}
	\end{subfigure}
	\begin{subfigure}[t]{\linewidth}
		\centering
		\includegraphics[width=0.8\linewidth]{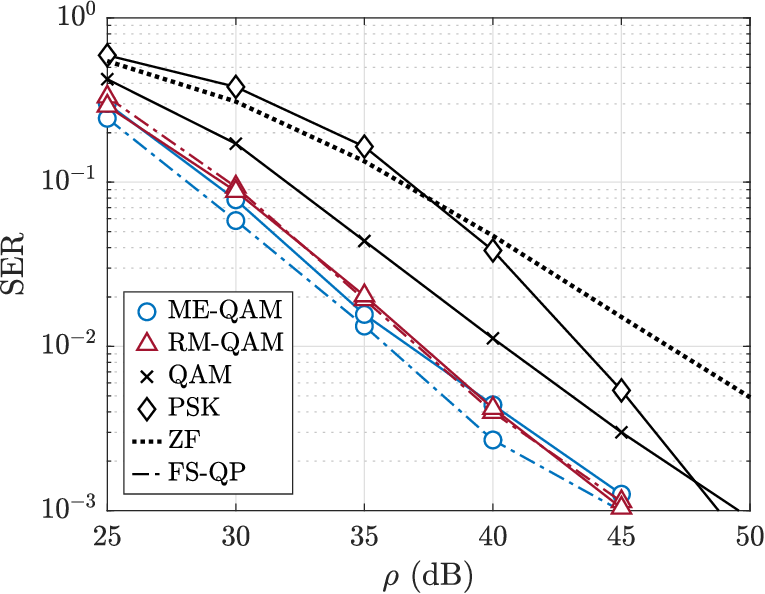}
		\caption{$M=64$}
		\label{fig:ser_k16_m64}
	\end{subfigure}
	\caption{SER versus $\rho$ for a $16\times16$ MIMO 
		system under perfect CSI, with (a) $M=16$ and (b) $M=64$. CIP 
		with ME-QAM and RM-QAM is compared against QAM- and PSK-based 
		CIP and ZF precoding with QAM. FS-QP results are included for 
		RM-QAM in both cases and for ME-QAM in (b), confirming the 
		near-optimality of PS-QP.}
	\label{fig:ser_k16}
\end{figure}

\begin{figure}[t]
	\centering
	\begin{subfigure}[t]{0.8\linewidth}
		\centering
		\includegraphics[width=\linewidth]{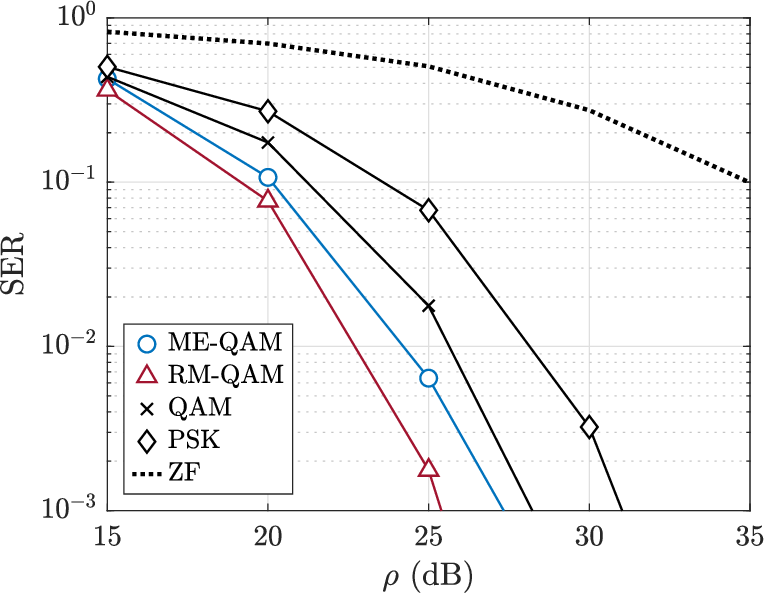}
		\caption{$M=16$}
		\label{fig:ser_k64_m16}
	\end{subfigure}
	\begin{subfigure}[t]{\linewidth}
		\centering
		\includegraphics[width=0.8\linewidth]{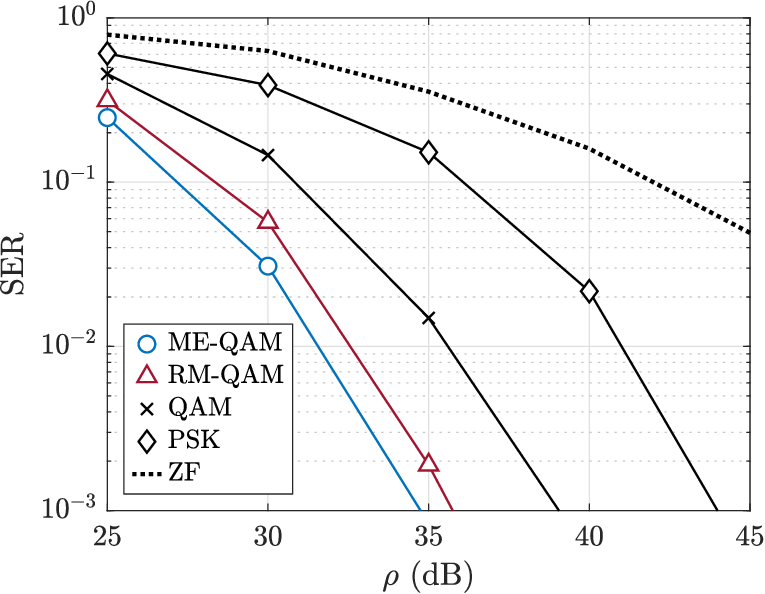}
		\caption{$M=64$}
		\label{fig:ser_k64_m64}
	\end{subfigure}
	\caption{SER versus $\rho$ for a $64\times64$ MIMO 
		system under perfect CSI, with (a) $M=16$ and (b) $M=64$. CIP 
		with ME-QAM and RM-QAM is compared against QAM- and PSK-based 
		CIP and ZF precoding with QAM.}
	\label{fig:ser_k64}
\end{figure}

\begin{figure}[t]
	\centering
	\begin{subfigure}[t]{0.8\linewidth}
		\centering
		\includegraphics[width=\linewidth]{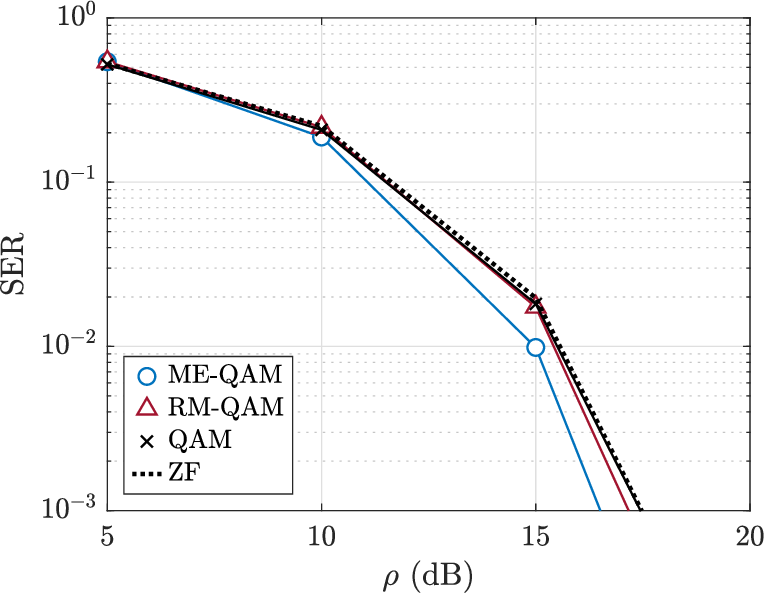}
		\caption{$M=16$}
		\label{fig:ser_k32_m16}
	\end{subfigure}
	\begin{subfigure}[t]{\linewidth}
		\centering
		\includegraphics[width=0.8\linewidth]{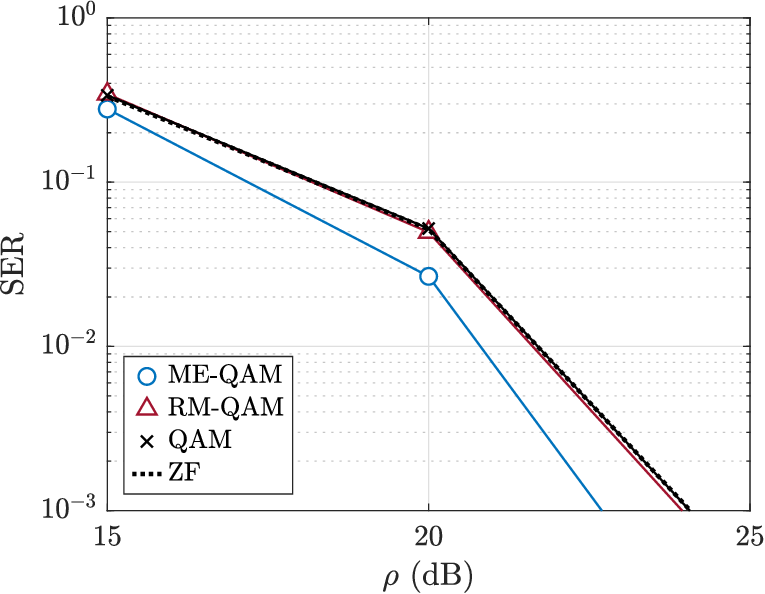}
		\caption{$M=64$}
		\label{fig:ser_k32_m64}
	\end{subfigure}
	\caption{SER versus $\rho$ for a $64\times32$ MIMO 
		system under perfect CSI, with (a) $M=16$ and (b) $M=64$. CIP 
		with ME-QAM and RM-QAM is compared against QAM-based 
		CIP and ZF precoding.}
	\label{fig:ser_k32}
\end{figure}

This experiment demonstrates the SER performance with respect to transmit SNR $\rho = 1/\sigma^2$ of the proposed schemes under different MIMO configurations.

Fig.~\ref{fig:ser_k16} presents results for a $16\times16$ MIMO system 
with $M\in\{16,64\}$. FS-QP results are included for RM-QAM in 
Fig.~\ref{fig:ser_k16_m16} and for both schemes in Fig.~\ref{fig:ser_k16_m64}, 
confirming the near-optimality of PS-QP. For $M=16$, RM-QAM achieves more than $4$~dB gain over all baselines at $\text{SER}=10^{-3}$, while ME-QAM exhibits diminishing gain over QAM at high SNR, consistent with the CCDF results in Fig.~\ref{fig:ccdf}. For $M=64$, 
ME-QAM and RM-QAM exhibit similar performance, both achieving 
approximately $4$~dB gain over the baselines.

Interestingly, PSK outperforms QAM in CIP at high SNR despite its 
higher symbol energy. We attribute this to the fact that every PSK 
symbol maps to a non-singleton feasible region, maintaining 
consistent feasible set dimensions across symbol realizations. In 
contrast, QAM occasionally produces symbol vectors $\mathbf{s}$ 
with severely limited feasibility, leading to large $\alpha^2$ and 
degraded SER. This effect is more pronounced in smaller MIMO 
systems where the available DoF are inherently 
limited.

Fig.~\ref{fig:ser_k64} extends the evaluation to a $64\times64$ MIMO 
system. The additional DoF benefit all CIP-based 
schemes. For $M=16$, RM-QAM achieves a $3$~dB SNR gain over QAM 
at $\text{SER}=10^{-3}$; for $M=64$, ME-QAM achieves a slightly 
higher gain of $4$~dB at the same SER level.

Fig.~\ref{fig:ser_k32} considers an underdetermined $64\times32$ 
($N_\mathrm{t}\times K$) MIMO configuration. In this regime, 
$(\dot{\mathbf{H}}\dot{\mathbf{H}}^\mathsf{T})^{-1}$ is close to 
a scaled identity, so the CIP objective approximates 
$\|\dot{\mathbf{s}}\|^2$ and is minimized by the smallest feasible 
$\dot{\mathbf{s}}$, driving most constraints to be active at the 
optimum and causing CIP to degenerate approximately to ZF 
precoding. This explains the nearly identical results between 
QAM-based CIP and ZF. In contrast, the 
sign flexibility of ME-QAM can still be exploited despite the 
active constraints, yielding approximately $1$~dB gain over the 
baselines. For RM-QAM, the performance gain from the free DoF is offset 
by the penalty in $E_s$, 
resulting in performance comparable to the baselines.

\subsection{Experiment 4: SER under Imperfect CSI}
\begin{figure}[t]
	\centering
	\begin{subfigure}[t]{0.8\linewidth}
		\centering
		\includegraphics[width=\linewidth]{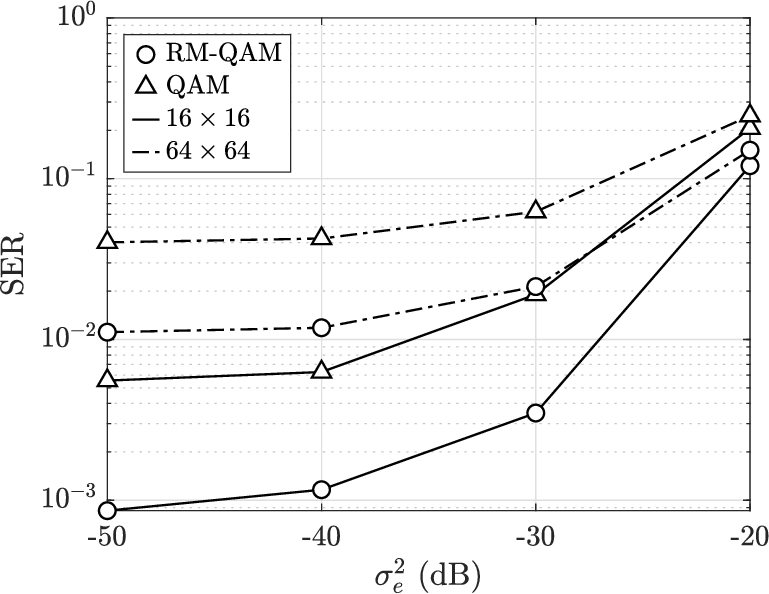}
		\caption{$M=16$}
		\label{fig:ser_icsi_m16}
	\end{subfigure}
	\begin{subfigure}[t]{\linewidth}
		\centering
		\includegraphics[width=0.8\linewidth]{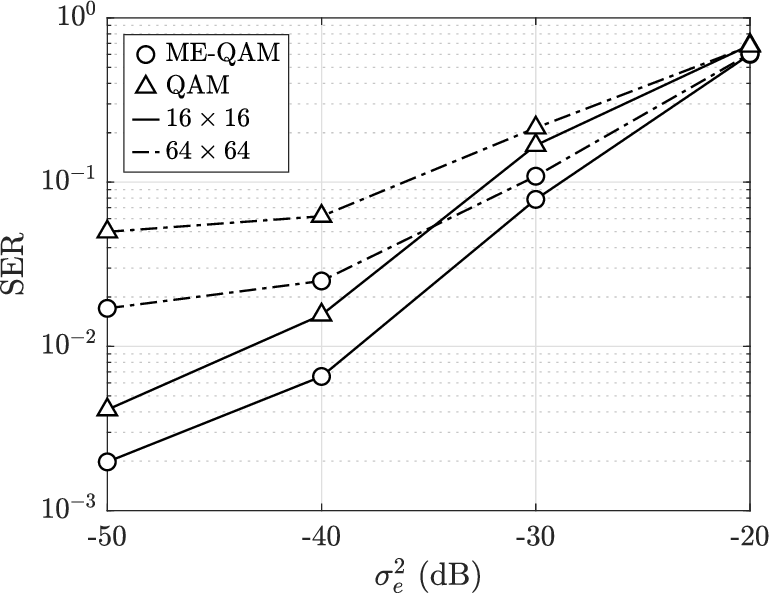}
		\caption{$M=64$}
		\label{fig:ser_icsi_m64}
	\end{subfigure}
	\caption{SER versus $\sigma_e^2$ for $16\times16$ 
		and $64\times64$ MIMO under imperfect CSI. (a) RM-QAM and (b) ME-QAM 
		are each compared against QAM-based CIP. The transmit SNR is set to 
		$\rho=30$\,dB and $\rho=25$\,dB for the $16\times16$ and $64\times64$ 
		systems with $M=16$, and $\rho=45$\,dB and $\rho=35$\,dB for $M=64$, 
		corresponding to an SER of approximately $10^{-3}$ under perfect CSI.}
	\label{fig:ser_icsi}
\end{figure}

This experiment examines whether the performance gains in SER are consistent under imperfect CSI. The channel 
estimate is modeled as $\hat{\mathbf{H}} = \mathbf{H} + \mathbf{E}$, 
with entries of $\mathbf{E}$ drawn i.i.d.\ from 
$\mathcal{CN}(0,\sigma_e^2)$, consistent with pilot-based estimation 
and widely adopted in the MIMO literature~\cite{Sohrabi2020, 
	Yoo2006, Wang2007}. All BS signal processing operates on 
$\hat{\mathbf{H}}$ in place of the true $\mathbf{H}$.

Fig.~\ref{fig:ser_icsi} presents the SER as a function of $\sigma_e^2$. 
The transmit SNR is set to target an SER of approximately $10^{-3}$ 
under perfect CSI: $\rho = 30$~dB ($16\times16$) and $\rho = 25$~dB 
($64\times64$) for $M=16$, and $\rho = 45$~dB ($16\times16$) and 
$\rho = 35$~dB ($64\times64$) for $M=64$. For clarity, only RM-QAM 
is compared against QAM for $M=16$, and ME-QAM against QAM for 
$M=64$, as each represents the better-performing scheme in each 
case under perfect CSI.

As expected, the SER of all schemes degrades monotonically with 
$\sigma_e^2$. The proposed schemes maintain a clear advantage over 
QAM at $-50$ to $-30$ dB error levels; however, the gap narrows as 
$\sigma_e^2$ approaches $-20$~dB, where channel estimation error 
becomes the dominant performance bottleneck for all schemes. 
The results indicate
that the performance advantage of the proposed schemes is robust to moderate CSI 
imperfections.
\begin{figure}[t]
	\centering
	\includegraphics[width=0.35\textwidth]{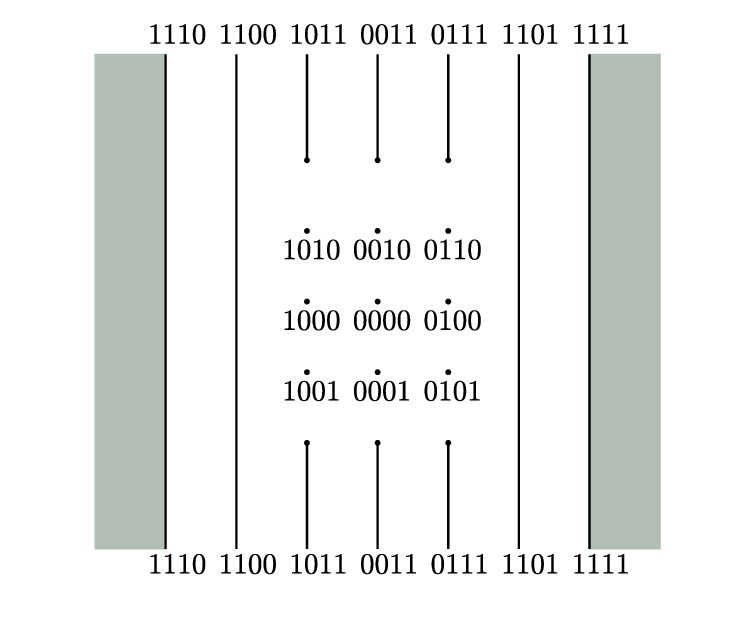}
	\vspace{-0.3cm}
	\caption{Bit mapping for $16$-ary RM-QAM with an average Hamming distance of $1.1$ per adjacent symbol pair.}
	\label{fig:bitmap}
\end{figure}

\subsection{Experiment 5: BLER}
\begin{figure}[t]
	\centering
	\begin{subfigure}[t]{0.8\linewidth}
		\centering
		\includegraphics[width=\linewidth]{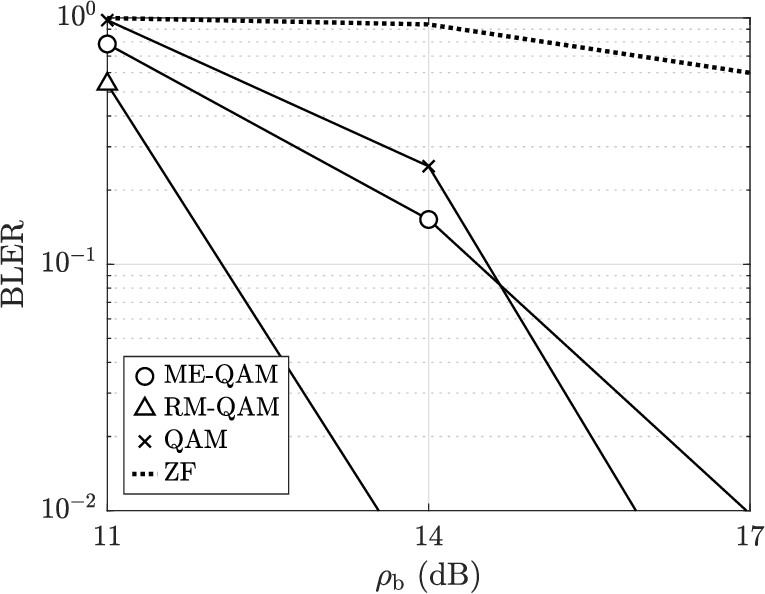}
		\caption{$M=16$}
		\label{fig:bler_m16}
	\end{subfigure}
	\begin{subfigure}[t]{\linewidth}
		\centering
		\includegraphics[width=0.8\linewidth]{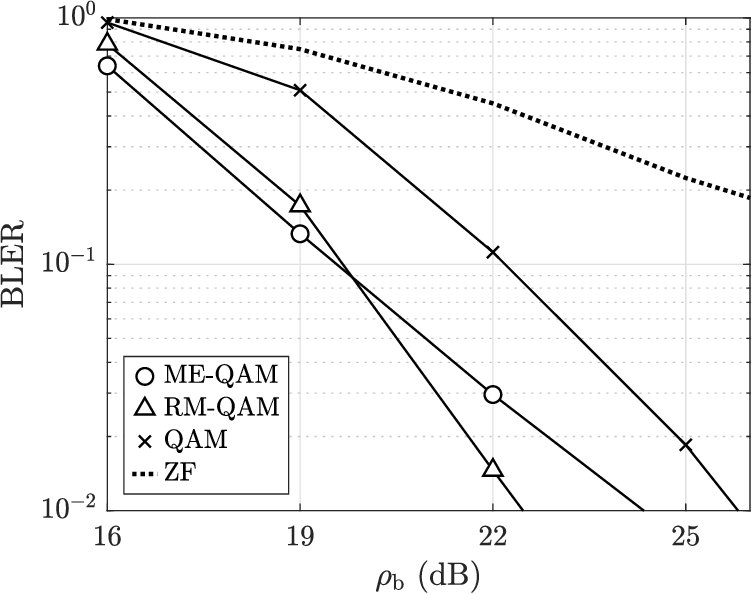}
		\caption{$M=64$}
		\label{fig:bler_m64}
	\end{subfigure}
	\caption{BLER versus transmit SNR per information bit $\rho_\mathrm{b}$ 
		for a $16\times16$ MIMO system with rate-$3/4$ LDPC coding, with (a) 
		$M=16$ and (b) $M=64$. CIP with ME-QAM and RM-QAM is compared against 
		QAM-based CIP and ZF precoding.}
	\label{fig:bler}
\end{figure}
This experiment examines the BLER performance of the proposed 
schemes to validate their compatibility with channel coding. A 
rate-$3/4$ LDPC code is employed, owing 
to its capacity-approaching performance and widespread adoption in 
modern wireless standards such as IEEE 802.11 and 5G NR.

Gray mapping is applied to all schemes, which minimizes the average 
Hamming distance between adjacent feasible regions~\cite{Agrell2004}. For ME-QAM, a perfect Gray mapping is 
achievable due to its cyclic structure in each real dimension, 
analogous to PSK. For RM-QAM, however, a perfect Gray mapping is not 
attainable since some regions have more nearest neighbors than the 
number of bits per symbol. Our mappings for 16- and 64-ary 
RM-QAM achieve average Hamming distances of approximately $1.1$ 
and $1.2$ bits per adjacent symbol pair, respectively. The bit 
mapping for 16-ary RM-QAM is illustrated in Fig.~\ref{fig:bitmap}.

Fig.~\ref{fig:bler} shows the BLER performance for the $16\times16$ 
system, where $\rho_\mathrm{b} = \rho/(R\log_2 M)$ denotes the 
transmit SNR per information bit with code rate $R=3/4$. For 
$M=16$, RM-QAM achieves approximately $2$~dB gain over QAM at 
$\text{BLER}=10^{-2}$, while ME-QAM underperforms QAM at low BLERs but recovers to achieve $1$~dB gain at $\text{BLER}=10^{-2}$.
For $M=64$, both schemes outperform QAM, with RM-QAM achieving a higher $3$~dB gain at $\text{BLER}=10^{-2}$.

\section{Conclusion}\label{sec6}
In this paper, we studied constellation design for CIP in MU-MIMO systems. By revisiting QAM-based CIP, we analytically showed that the misalignment between the CI regions and the objective-minimizing sign pattern fundamentally limits performance, and that introducing SF regions can significantly alleviate this limitation. Based on the analysis, we proposed the RBC model to lift the restrictions in the conventional CI region model. ME-QAM and RM-QAM RBC schemes with improved sign-alignment capability were developed and shown by simulations to achieve superior performance over QAM in CIP. Additionally, the PS-QP algorithm was developed to solve the resulting MIQP with complexity comparable to QAM-based CIP, while achieving near-identical performance to the optimal FS-QP algorithm. Given the advantage of the RBC model in CIP, its extension to other systems with non-bijective modulation represents a promising direction for future research.
\appendices
\section{Proof of \textit{Proposition}~\ref{prop1}}\label{apl1}
Each row of $\dot{\mathbf{H}}_{\mathcal{I}_\mathrm{in}}$ is in the form of either $[\Re(\mathbf{h}_k^\mathsf{T}), -\Im(\mathbf{h}_k^\mathsf{T})]$ or $[\Im(\mathbf{h}_k^\mathsf{T}), \Re(\mathbf{h}_k^\mathsf{T})]$ for some user $k$. By~\cite{deluna2018realcomplex}, the diagonal element of $(\dot{\mathbf{H}}_{\mathcal{I}_\mathrm{in}}\dot{\mathbf{H}}_{\mathcal{I}_\mathrm{in}}^\mathsf{T})^{-1}$ follows a scaled inverse-$\chi^2$ distribution with scale 2 and with DoF depending on whether the corresponding row has its pair also retained in $\dot{\mathbf{H}}_{\mathcal{I}_\mathrm{in}}$. A row whose paired row is absent corresponds to $2N_\mathrm{t}-|\mathcal{I}_\mathrm{in}|+1$ DoF. Otherwise, the DoF is $2N_\mathrm{t}-|\mathcal{I}_\mathrm{in}|+2$. Acknowledging that the expected value of the scaled inverse-$\chi^2$ variable 
with $\nu>2$ DoF is $2/(\nu-2)$, the expected value of each diagonal element in $(\dot{\mathbf{H}}_{\mathcal{I}_\mathrm{in}}\dot{\mathbf{H}}_{\mathcal{I}_\mathrm{in}}^\mathsf{T})^{-1}$ is lower-bounded by $2/(2N_\mathrm{t}-|\mathcal{I}_\mathrm{in}|)$. Therefore, we obtain the following bound for a given $|\mathcal{I}_\mathrm{in}|\le2N_\mathrm{t}-2$:
\begin{equation}\label{eq:approx_ex}
	\begin{aligned}
		\mathbb{E}\big(\mathrm{tr}\big((\dot{\mathbf{H}}_{\mathcal{I}_\mathrm{in}}
		\dot{\mathbf{H}}_{\mathcal{I}_\mathrm{in}}^\mathsf{T})^{-1}\big)\big)
		\big|_{|\mathcal{I}_\mathrm{in}|}
		&= \mathrm{tr}\big(\ex\big((\dot{\mathbf{H}}_{\mathcal{I}_\mathrm{in}}\dot{\mathbf{H}}_{\mathcal{I}_\mathrm{in}}^\mathsf{T})^{-1}\big|_{|\mathcal{I}_\mathrm{in}|}\big)\big)\\
		&\ge \frac{2|\mathcal{I}_\mathrm{in}|}{2N_\mathrm{t}-|\mathcal{I}_\mathrm{in}|}.
	\end{aligned}
\end{equation}
Substituting \eqref{eq:approx_ex} into \eqref{eq:ap} gives
\begin{equation}\label{eq:a2}
	\begin{aligned}
		\mathbb{E}(\alpha'^{2}) &\ge E_{s,\mathrm{in}}\sum_{n=0}^{\min\{2K,2N_\mathrm{t}-2\}}\pr\{|\ids_\mathrm{in}|=n\}\frac{2n}{2N_\mathrm{t}-n}\\
		&\ge E_{s,\mathrm{in}}\frac{2(\overline{|\ids_\mathrm{in}|}-\epsilon)}{2N_\mathrm{t}-(\overline{|\ids_\mathrm{in}|}-\epsilon)}= \frac{2E_{s,\mathrm{in}}}{2N_\mathrm{t}/(\overline{|\ids_\mathrm{in}|}-\epsilon)-1}
	\end{aligned}
\end{equation}
where $\overline{|\mathcal{I}_\mathrm{in}|} = \mathbb{E}(|\mathcal{I}_\mathrm{in}|)$ 
is the expected number of interior symbols in $\dot{\mathbf{s}}$, 
$\epsilon \ge 0$ accounts for the probability mass on the excluded 
region $|\mathcal{I}_\mathrm{in}| > 2N_\mathrm{t}-2$, and the second 
step follows from Jensen's inequality since $2n/(2N_\mathrm{t}-n)$ 
is convex in $n$. According to \eqref{eq:dpci_pam}, we have $|\mathcal{I}_\mathrm{in}| \sim 
\mathrm{Binomial}(2K, (L-2)/L)$. Therefore, 
\begin{equation}\label{eq:num_in}
	\overline{|\mathcal{I}_\mathrm{in}|} = 2K(L-2)/{L}.
\end{equation}

When $K \le N_\mathrm{t}-1$, since $|\mathcal{I}_\mathrm{in}| \le 2K 
\le 2N_\mathrm{t}-2$ always holds, the sum is untruncated and 
$\epsilon = 0$ exactly. When $K = N_\mathrm{t}$, 
$\Pr\{|\mathcal{I}_\mathrm{in}| > 2N_\mathrm{t}-2\} \to 0$ as 
$K \to \infty$ by concentration of the binomial distribution, 
and therefore $\epsilon \to 0$. Substituting \eqref{eq:num_in} 
into \eqref{eq:a2}, together with $\epsilon = 0$ or $\epsilon \to 0$ 
from the above, yields \eqref{eq:ap_lb}. This completes the proof.

\section{Proof of \textit{Proposition}~\ref{prop2}}\label{apl2}
It suffices to consider a single $i \in \{1,\cdots,|\mathcal{I}_{\mathrm{end}}|\}$, 
as the argument extends identically to all entries, regardless of 
whether $i$ corresponds to the real or imaginary part of an entry 
of $\mathbf{s}'$ (i.e., the complex-valued version of $\dot{\mathbf{s}}'$). Without loss of generality, we assume $i$ 
corresponds to $\Re(s'_k)$ for some $k\in\{1,2,\cdots,K\}$. Let 
$\mathbf{D}\in \mathbb{C}^{K\times K}$ be the diagonal matrix with 
all diagonal entries equal to $1$ except the $k$th equal to 
$-s_k^{'*}/s'_k$, which satisfies $(-s_k^{'*}/s'_k)\cdot s'_k = 
-\Re(s'_k)+j\Im(s'_k)$, i.e., $\mathbf{D}$ flips $z'_i$ while 
leaving all other entries of $\mathbf{z}'$ unchanged. Since 
$|-s_k^{'*}/s'_k|=1$, $\mathbf{D}$ is unitary, and therefore
\begin{equation}
	\mathbf{s}'^{\mathsf{H}}(\mathbf{H}\mathbf{H}^\mathsf{H})^{-1}\mathbf{s}' 
	= (\mathbf{D}\mathbf{s}')^\mathsf{H}
	\big(\mathbf{D}\mathbf{H}(\mathbf{D}\mathbf{H})^\mathsf{H}\big)^{-1}
	\mathbf{D}\mathbf{s}',
\end{equation}
which indicates that $\mathbf{D}\mathbf{s}'$ is an optimal solution for the channel 
$\mathbf{D}\mathbf{H}$. Since $\mathbf{D}$ applies only a phase 
rotation to the $k$th row of $\mathbf{H}$, and the entries of 
$\mathbf{H}$ are i.i.d.\ with a distribution symmetric under sign 
inversion and complex conjugation, $\mathbf{D}\mathbf{H}$ and 
$\mathbf{H}$ are identically distributed. Consequently, $\mathbf{s}'$ 
and $\mathbf{D}\mathbf{s}'$ are equiprobable optimal solutions, and 
since they differ only in $z'_i$, we have $\Pr(z'_i=+1)=\Pr(z'_i=-1)=\tfrac{1}{2}$.

Since $z_i = \mathrm{sgn}(\Re(s_k))$ is determined solely by $m_k$, which is independent of the random variables that determine $z'_i$, $z_i$ is independent of $z'_i$. Therefore,
\begin{equation}
	\begin{aligned}
		\Pr(z'_i = z_i) &= \Pr(z'_i=+1)\Pr(z_i=+1) 
		\\&\quad+ \Pr(z'_i=-1)\Pr(z_i=-1)= \tfrac{1}{2}.
	\end{aligned}
\end{equation}
This completes the proof.

\section{SER Upper Bound of ME-QAM}\label{apme}
We prove that $P_\mathrm{e}$ of ME-QAM under detection based on the decision boundaries in Fig.~\ref{fig:me} satisfies the SER upper bound in~\eqref{eq:ub}. 

For $M$-ary ME-QAM, the real and imaginary parts of $s$ are detected independently as $L$-ary ME-PAM symbols. Let $\dot{s}$, $\dot{y}$, and $\dot{v}$ denote the real part of $s$, $\bar{y}$, and $v$, respectively. Since $v \sim \mathcal{CN}(0, \sigma^2)$, we have $\dot{v} \sim \mathcal{N}(0, \sigma^2/2)$, and thus $\alpha\dot{v} \sim \mathcal{N}(0, \bar{\sigma}^2/2)$. The decision boundaries in Fig.~\ref{fig:me} correspond to $\mathcal{B} = \{-L+1, -L+3, \dots, L-1\}$ in each real dimension. Denote the error probability per real dimension as $\dot{P}_{\mathrm{e}}$. For each symbol of the first case in \eqref{eq:mepam}, the nearest decision boundaries are at distance $d_{\min}/2 = 1$, giving~\cite{Proakis2008}
\begin{equation}
	\dot{P}_{\mathrm{e}}\big|_{\ell \in \{0,\dots,L-2\}} =2Q\bigg(\frac{1}{\sqrt{\bar{\sigma}^2/2}}\bigg)=  2Q\bigg(\sqrt{\frac{2}{\bar{\sigma}^2}}\bigg).
\end{equation}
For $\ell=L-1$, the nearest decision boundary $\operatorname{sgn}(\dot{s})(L-1)$ is no farther from any $\dot{s} \in \mathcal{R}(L-1)$ than from the boundary point $\operatorname{sgn}(\dot{s})L$. Therefore,
\begin{equation}
	\begin{aligned}
		&\dot{P}_{\mathrm{e}}\big|_{\ell=L-1} \le \dot{P}_{\mathrm{e}}\big|_{\ell=L-1,~\dot{s}=\operatorname{sgn}(\dot{s})L} \\
		&= Q\bigg(\frac{1}{\sqrt{\bar{\sigma}^2/2}}\bigg) - Q\bigg(\frac{2L-1}{\sqrt{\bar{\sigma}^2/2}}\bigg) \le Q\bigg(\sqrt{\frac{2}{\bar{\sigma}^2}}\bigg).
	\end{aligned}
\end{equation}
Letting $Q_0 = Q\big(\sqrt{2/\bar{\sigma}^2}\big)$, we obtain
\begin{equation}
	\dot{P}_{\mathrm{e}} \le \frac{L-1}{L} 2Q_0 + \frac{1}{L} Q_0 = \frac{2L-1}{L}Q_0.
\end{equation}
Since an ME-QAM symbol is correctly detected only when both its real and imaginary parts are correctly detected and both parts are independent, the overall SER is upper-bounded as
\begin{equation}
	P_{\mathrm{e}} = 1 - (1 - \dot{P}_{\mathrm{e}})^2 \le 2\dot{P}_{\mathrm{e}} \le \frac{4L-2}{L} Q_0.
\end{equation}
For any $M \ge 4$ (i.e., $L \ge 2$), $P_{\mathrm{e}}$ satisfies the upper bound in \eqref{eq:ub}, which covers the constellation sizes of interest.

\section{SER Upper Bound of RM-QAM}\label{aprm}
We prove that $P_\mathrm{e}$ of RM-QAM under detection based on the decision boundaries illustrated in Fig.~\ref{fig:rm} satisfies the upper bound in \eqref{eq:ub}.

The decision boundaries separate the detection of the real and imaginary parts. Let $P_{\mathfrak{R}}$ and $P_{\mathfrak{I}}$ denote the error probabilities of the real and imaginary parts, respectively. By arguments similar to those in Appendix~\ref{apme},
\begin{equation}
	\begin{aligned}
		&P_{\mathfrak{R}}\big|_{m\in\mathcal{M}_1\cup\mathcal{M}_2\cup\mathcal{M}_3} = 2Q_0, \quad P_{\mathfrak{R}}\big|_{m\in\mathcal{M}_4} \le Q_0,\\
		&P_{\mathfrak{I}}\big|_{m\in\mathcal{M}_1} = 2Q_0, \quad P_{\mathfrak{I}}\big|_{m\in\mathcal{M}_2} \le Q_0, \quad P_{\mathfrak{I}}\big|_{m\in\mathcal{M}_3\cup\mathcal{M}_4} = 0,
	\end{aligned}
\end{equation}
where the last equation is due to the fact that symbols in these groups are fully determined by their real part. Since a symbol is in error when at least one of its real and imaginary parts is incorrectly detected, the per-group error probabilities are bounded via the union bound, regardless of the dependence between the two parts, as
\begin{equation}
	\begin{aligned}
		&P_{\mathrm{e}}\big|_{m\in\mathcal{M}_1} \le P_{\mathfrak{R}}\big|_{m\in\mathcal{M}_1} + P_{\mathfrak{I}}\big|_{m\in\mathcal{M}_1} = 4Q_0,\\
		&P_{\mathrm{e}}\big|_{m\in\mathcal{M}_2} \le P_{\mathfrak{R}}\big|_{m\in\mathcal{M}_2} + P_{\mathfrak{I}}\big|_{m\in\mathcal{M}_2} \le 3Q_0,\\
		&P_{\mathrm{e}}\big|_{m\in\mathcal{M}_3} = P_{\mathfrak{R}}\big|_{m\in\mathcal{M}_3} = 2Q_0,\\
		&P_{\mathrm{e}}\big|_{m\in\mathcal{M}_4} = P_{\mathfrak{R}}\big|_{m\in\mathcal{M}_4} \le Q_0.
	\end{aligned}
\end{equation}
Therefore,
\begin{equation}
	\begin{aligned}
		P_{\mathrm{e}} &= \sum_{i=1}^{4}\frac{|\mathcal{M}_i|}{M}\, P_{\mathrm{e}}'\big|_{m\in\mathcal{M}_i}\\
		&\le \frac{Q_0}{L^2}\big(4(L-1) + 3(L-1) + (L-2) + 2\big)\\
		&= \frac{4L^2-3L+3}{L^2}Q_0.
	\end{aligned}
\end{equation}
For any $L \ge 4$, $P_{\mathrm{e}}$ satisfies the upper bound in \eqref{eq:ub}, which covers the constellation sizes of interest.

\section{Proof of \textit{Proposition}~\ref{prop4}}\label{apl3}
Let $\alpha^{2}$ denote the objective value of \eqref{eq:cipme} corresponding to the optimal solution $\dot{\mathbf{s}}^\star$. Let $\mathcal{I}^\mathsf{c}_B$ denote the complement of $\mathcal{I}_B$ in $\{1,2,\ldots,2K\}$ and $\alpha'^{2}$ denote the optimal objective of the \eqref{eq:cipme} with the constraints associated with $\mathcal{I}_B$ removed, so that $\Delta = \alpha^{2} - \alpha'^{2}$. Consider the auxiliary problem obtained by fixing the entries indexed by $\mathcal{I}^\mathsf{c}_B$ at $\dot{\mathbf{s}}_{\mathcal{I}^\mathsf{c}_B}$ and minimizing freely over $\mathbf{u} \in \mathbb{R}^{|\mathcal{I}_B|}$:
\begin{equation}
	J
	=
	\min_{\mathbf{u}\in\mathbb{R}^{|\mathcal{I}_B|}}
	\begin{bmatrix}
		\dot{\mathbf{s}}^\star_{\mathcal{I}^\mathsf{c}_B}\\
		\mathbf{u}
	\end{bmatrix}^{\!\mathsf{T}}
	\breve{\mathbf{Q}}
	\begin{bmatrix}
		\dot{\mathbf{s}}^\star_{\mathcal{I}^\mathsf{c}_B}\\
		\mathbf{u}
	\end{bmatrix},
\end{equation}
where 
\begin{equation}
	\breve{\mathbf{Q}} =
	\begin{bmatrix}
		\mathbf{Q}_{\mathcal{I}^\mathsf{c}_B\mathcal{I}^\mathsf{c}_B}
		&
		\mathbf{Q}_{\mathcal{I}^\mathsf{c}_B\mathcal{I}_B}
		\\
		\mathbf{Q}_{\mathcal{I}_B\mathcal{I}^\mathsf{c}_B}
		&
		\mathbf{Q}_{\mathcal{I}_B\mathcal{I}_B}
	\end{bmatrix}.
\end{equation}
Minimizing the quadratic form over $\mathbf{u}$ yields the optimal auxiliary solution
\begin{equation}
	\mathbf{u}
	=
	-\mathbf{Q}_{\mathcal{I}_B\mathcal{I}_B}^{-1}
	\mathbf{Q}_{\mathcal{I}_B\mathcal{I}^\mathsf{c}_B}
	\dot{\mathbf{s}}^\star_{\mathcal{I}^\mathsf{c}_B},
\end{equation}
and $J$ is given by the corresponding Schur complement~\cite{Boyd2004}:
\begin{equation}
	J
	=
	\dot{\mathbf{s}}_{\mathcal{I}^\mathsf{c}_B}^\mathsf{T}
	\big(
	\mathbf{Q}_{\mathcal{I}^\mathsf{c}_B\mathcal{I}^\mathsf{c}_B}
	-\mathbf{Q}_{\mathcal{I}^\mathsf{c}_B\mathcal{I}_B}
	\mathbf{Q}_{\mathcal{I}_B\mathcal{I}_B}^{-1}
	\mathbf{Q}_{\mathcal{I}_B\mathcal{I}^\mathsf{c}_B}
	\big)
	\dot{\mathbf{s}}_{\mathcal{I}^\mathsf{c}_B}.
\end{equation}
Expanding $\alpha^2=\dot{\mathbf{s}}^{\star\tp}\mathbf{Q}\dot{\mathbf{s}}^\star$ under the same block partition yields
\begin{equation}
	\alpha^2=
	\dot{\mathbf{s}}_{\mathcal{I}^\mathsf{c}_B}^{\star\mathsf{T}}
	\mathbf{Q}_{\mathcal{I}^\mathsf{c}_B\mathcal{I}^\mathsf{c}_B}
	\dot{\mathbf{s}}_{\mathcal{I}^\mathsf{c}_B}^\star
	+
	2\dot{\mathbf{s}}_{\mathcal{I}^\mathsf{c}_B}^{\star\mathsf{T}}
	\mathbf{Q}_{\mathcal{I}^\mathsf{c}_B\mathcal{I}_B}
	\dot{\mathbf{s}}^\star_{\mathcal{I}_B}
	+
	\dot{\mathbf{s}}_{\mathcal{I}_B}^{\star\mathsf{T}}
	\mathbf{Q}_{\mathcal{I}_B\mathcal{I}_B}
	\dot{\mathbf{s}}^\star_{\mathcal{I}_B}.
\end{equation}
Therefore,
\begin{equation}
	\begin{aligned}
		&\alpha^2 - J\\
		=&
		\dot{\mathbf{s}}_{\mathcal{I}^\mathsf{c}_B}^\mathsf{T}
		\mathbf{Q}_{\mathcal{I}^\mathsf{c}_B\mathcal{I}_B}
		\mathbf{Q}_{\mathcal{I}_B\mathcal{I}_B}^{-1}
		\mathbf{Q}_{\mathcal{I}_B\mathcal{I}^\mathsf{c}_B}
		\dot{\mathbf{s}}_{\mathcal{I}^\mathsf{c}_B}
		+2\dot{\mathbf{s}}_{\mathcal{I}^\mathsf{c}_B}^\mathsf{T}
		\mathbf{Q}_{\mathcal{I}^\mathsf{c}_B\mathcal{I}_B}
		\dot{\mathbf{s}}_{\mathcal{I}_B}
		\\&+\dot{\mathbf{s}}_{\mathcal{I}_B}^\mathsf{T}
		\mathbf{Q}_{\mathcal{I}_B\mathcal{I}_B}
		\dot{\mathbf{s}}_{\mathcal{I}_B} \\
		=&
		\Big(
		\mathbf{Q}_{\mathcal{I}_B\mathcal{I}^\mathsf{c}_B}\dot{\mathbf{s}}_{\mathcal{I}^\mathsf{c}_B}
		+
		\mathbf{Q}_{\mathcal{I}_B\mathcal{I}_B}\dot{\mathbf{s}}_{\mathcal{I}_B}
		\Big)^\mathsf{T}
		\mathbf{Q}_{\mathcal{I}_B\mathcal{I}_B}^{-1}
		\Big(
		\mathbf{Q}_{\mathcal{I}_B\mathcal{I}^\mathsf{c}_B}\dot{\mathbf{s}}_{\mathcal{I}^\mathsf{c}_B}
		\\&+
		\mathbf{Q}_{\mathcal{I}_B\mathcal{I}_B}\dot{\mathbf{s}}_{\mathcal{I}_B}
		\Big) \\
		=&\;
		\mathbf{g}_{\mathcal{I}_B}^\mathsf{T}
		\mathbf{Q}_{\mathcal{I}_B\mathcal{I}_B}^{-1}
		\mathbf{g}_{\mathcal{I}_B} = \underline{\Delta}.
	\end{aligned}
\end{equation}
Since fixing $\dot{\mathbf{s}}_{\mathcal{I}^\mathsf{c}_B}$ restricts the feasible set of the fully relaxed problem, we have $J \geq \alpha'^2$. Therefore,
\begin{equation}
	\Delta
	=
	\alpha^2 - \alpha'^2
	\;\ge\;
	\alpha^2 - J
	=
	\underline{\Delta}.
\end{equation}
This completes the proof.
\balance
\bibliographystyle{IEEEtran}
\bibliography{rbc_ref}

\end{document}